\documentclass[runningheads]{llncs}

\usepackage{cite}
\usepackage{hyperref}
\usepackage{colortbl}
\usepackage{paralist}
\usepackage{pdfpages}
\usepackage{enumitem}
\usepackage{float}
\usepackage{booktabs}
\usepackage[override]{cmtt}
\usepackage{color,xcolor}
\usepackage{graphicx,eso-pic}
\usepackage{amssymb}
\usepackage{tikz}
\usepackage{tikzscale}
\usetikzlibrary{arrows.meta, decorations.pathmorphing}
\usetikzlibrary{shapes}

\newcommand{\N}{\mathbb{N}}

\newcommand{\Z}{\mathbb{Z}}

\newcommand{\Read}{{\ensuremath{\mathsf{read}}}\xspace}
\newcommand{\Write}{{\ensuremath{\mathsf{write}}}\xspace}
\newcommand{\pos}{{\ensuremath{\mathsf{pos}}}\xspace}
\newcommand{\dpos}{{\ensuremath{\mathsf{dpos}}}\xspace}

\newcommand{\key}{{\ensuremath{\mathsf{key}}}\xspace}

\newcommand{\ra}{\rightarrow}

\newcommand{\iperm}{\ensuremath{\mathfrak{e}}}

\newcommand{\zthree}{\ensuremath{\Z_3}}

\newcommand{\msf}[1]{\ensuremath{{\mathsf {#1}}}}
\newcommand{\mcal}[1]{\ensuremath{\mathcal {#1}}}

\newcommand{\algA}{{\ensuremath{\mcal{A}}}\xspace}
\newcommand{\adv}{{\ensuremath{\mcal{A}}}\xspace}
\newcommand{\Sim}{\ensuremath{\msf{Sim}}\xspace}

\definecolor{darkgreen}{rgb}{0,0.5,0}
\definecolor{lightblue}{RGB}{0,176,240}
\definecolor{darkblue}{RGB}{0,112,192}
\definecolor{lightpurple}{RGB}{124, 66, 168}
\definecolor{grey}{RGB}{139, 137, 137}
\definecolor{maroon}{RGB}{178, 34, 34}
\definecolor{green}{RGB}{34, 139, 34}
\definecolor{types}{RGB}{72, 61, 139}
\definecolor{gold}{rgb}{0.8, 0.33, 0.0}

\definecolor{darkgray}{gray}{0.3}

\newcommand{\skiptext}[1]{}
\newcommand{\Arr}{\ensuremath{{{\sf T}}}\xspace}

\renewcommand{\S}{\ensuremath{{{\bf S}}}\xspace}

\newcommand{\link}{\ensuremath{{{\sf L}}}\xspace}
\newcommand{\U}{\ensuremath{{{\sf U}}}\xspace}

\renewcommand{\P}{\ensuremath{{{\bf P}}}\xspace}

\newcommand{\algF}{\ensuremath{{{\mcal{F}}}}\xspace}
\renewcommand{\L}{\ensuremath{{{\mcal{L}}}}\xspace}
\newcommand{\C}{\ensuremath{{{\bf C}}}\xspace}
\newcommand{\view}{\ensuremath{{{\sf view}}}\xspace}

\newcommand{\Inp}{\ensuremath{{{\bf I}}}\xspace}
\newcommand{\Out}{\ensuremath{{{\bf O}}}\xspace}

\newenvironment{boxfig}[2]{\begin{figure}[#1]\fbox{\begin{minipage}{.95\columnwidth}
                        \vspace{0.2em}
                        \makebox[0.025\columnwidth]{}
                        \begin{minipage}{0.95\columnwidth}
            {\small{
                        #2 }}
                        \end{minipage}
                        \vspace{0.2em}
                        \end{minipage}}}{\end{figure}}

\usepackage{boxedminipage}
\usepackage{url}
\usepackage{caption}
\usepackage{subcaption}
\usepackage{xspace}
\usepackage{multirow}
\usepackage{amsmath,amstext,amssymb,amsfonts,latexsym}
\usepackage{wrapfig}
\usepackage{algorithm}
\usepackage[noend]{algpseudocode}
\usepackage{algorithmicx}
\usepackage{epstopdf}
\usepackage{mdframed}

\definecolor{darkred}{rgb}{0.5, 0, 0}
\definecolor{darkgreen}{rgb}{0, 0.5, 0}
\definecolor{darkblue}{rgb}{0,0,0.5}

\newcommand\markx[2]{}

\newcommand{\op}{\ensuremath{\mathsf{op}}\xspace}

\renewcommand{\path}{\ensuremath{\mathsf{path}}\xspace}

\newcommand{\data}{{\ensuremath{\mathsf{data}}}\xspace}
\newcommand{\addr}{{\ensuremath{\mathsf{addr}}}\xspace}
\newcommand{\addrd}[1]{{\ensuremath{\mathsf{addr}^{\langle #1 \rangle}}}\xspace}
\newcommand{\addrsd}[2]{{\ensuremath{\mathsf{addr}_{#1}^{\langle #2 \rangle}}}\xspace}
\newcommand{\Ud}[1]{{\ensuremath{U^{#1}}}\xspace}

\newcommand{\ORAM}{{\sf ORAM}}

\newcommand{\OTM}{{\sf OTM}}

\newcommand{\ignore}[1]{}

\newcounter{task}

\newtheorem{fact}[theorem]{Fact}

\newcommand{\kartik}[1]{{\footnotesize\color{blue}[kartik: #1]}}
\newcommand{\elaine}[1]{{\footnotesize\color{magenta}[Elaine: #1]}}
\newcommand{\anti}[1]{{\footnotesize\color{magenta}[Antigoni: #1]}}
\newcommand{\hubert}[1]{{\footnotesize\color{red}[Hubert: #1]}}

\renewcommand{\anti}[1]{}
\renewcommand{\elaine}[1]{}
\renewcommand{\hubert}[1]{}
\renewcommand{\kartik}[1]{}

\newcounter{cnt:challenge}

\definecolor{orange}{RGB}{255,70,0}
\newcommand{\TOTM}{{\sf OTM}}

\titlerunning{}

\begin{document}

\title{More is Less: Perfectly Secure
Oblivious Algorithms in the Multi-Server Setting}
\author{T-H.\ Hubert Chan\inst{1} \and
Jonathan Katz\inst{2} \and
Kartik Nayak\inst{2,3} \and
Antigoni Polychroniadou\inst{4} \and
Elaine Shi\inst{5}
}
%
%
\institute{The University of Hong Kong,
\email{hubert@cs.hku.hk}\\
\and
University of Maryland, College Park, \email{jkatz@cs.umd.edu}\\
\and
VMware Research, \email{nkartik@vmware.com}\\
\and
Cornell Tech, \email{antigoni@cornell.edu}\\
\and
Cornell University, \email{runting@gmail.com}
}
\maketitle

\begin{abstract}
The problem of Oblivious RAM (ORAM) has traditionally been
studied in a single-server setting, but more recently the
multi-server setting has also been considered. Yet it is still
unclear whether the multi-server setting has any
\emph{inherent} advantages, e.g., whether the multi-server
setting can be used to achieve stronger security goals or
provably better efficiency than is possible in the
single-server case.

\smallskip
In this work, we construct a perfectly secure 3-server ORAM
scheme that outperforms the best known single-server scheme by
a logarithmic factor. In the process we also show, for the
first time, that there exist specific algorithms for which
multiple servers can overcome known lower bounds in the
single-server setting.


\end{abstract}



\section{Introduction}
Oblivious RAM (ORAM) protocols~\cite{oram10} allow a client to
outsource storage of its data such that the client can continue
to read/write its data while hiding both the data itself as
well as the client's access pattern. ORAM was historically
considering in a single-server setting, but has recently been
considered in a multi-server
setting~\cite{multiserveroram,multicloudoram,abraham2017asymptotically,hoang2017s,cryptoeprint:2018:005,
DBLP:journals/corr/abs-1802-05145} where the client can store
its data on multiple, non-colluding servers. Current
constructions of multi-server ORAM are more efficient than
known protocols in the single-server setting; in particular,
the best known protocols in the latter setting (when
server-side computation is not allowed) require bandwidth
$O(\log^2 N/\log\log
N)$~\cite{oram09,oram03,oblivhash,circuitopram} for storing an
array of length~$N$, whereas multi-server ORAM schemes
achieving logarithmic bandwidth\footnote{Although Lu and
Ostrovsky~\cite{multiserveroram} describe their multi-server
scheme using server-side computation, it is not difficult to
see that it can be replaced with client-side computation
instead.} are known~\cite{multiserveroram}.


Nevertheless, there remain several unanswered questions about
the multi-server setting. First, all work thus far in the
multi-server setting achieves either computational or
statistical security, but not {\it perfect security} where
correctness is required to hold with probability~1 and security
must hold even against computationally unbounded attackers.
Second, although (as noted above) we have examples of
multi-server schemes that beat existing single-server
constructions, it is unclear whether this reflects a limitation
of our current knowledge or whether there are \emph{inherent}
advantages to the multi-server setting.

We address the above questions in this work. (Unless otherwise
noted, our results hold for arbitrary block size~$B$ as long as
it is large enough to store an address, i.e., $B = \Omega(\log
N)$.) First, we construct perfectly secure, multi-server ORAM
scheme that improves upon the overhead of the best known
construction in the single-server setting. Specifically, we
show:

\begin{theorem}
There exists a 3-server ORAM scheme that is perfectly secure
for any single semi-honest corruption, and achieves $O(\log^2
N)$ bandwidth per logical memory access on an array of
length~$N$. Further, our scheme does not rely on server-side
computation. \label{thm:intromain}
\end{theorem}
As a point of comparison, the best known \emph{single-server},
perfectly secure ORAM schemes require $O(\log^3 N)$
bandwidth~\cite{oram05,perfectopram}. While
Theorem~\ref{thm:intromain} holds for any block size
$B=\Omega(\log N)$, we show that for block sizes $B =
\Omega(\log^2 N)$
our scheme achieves bandwidth as small as~$O(\log N)$.


As part of our construction, we introduce new building blocks
that are of independent theoretical interest. Specifically, we
show:


\begin{theorem} 
There exists a 3-server protocol for stable compaction that is
perfectly secure for any single semi-honest corruption, and
achieves $O(N)$ bandwidth to compact an array of length~$N$
(that is secret shared among the servers). The same result
holds for merging two sorted arrays of length~$N$.
\end{theorem}

In the single-server setting, Lin, Shi, and
Xie~\cite{cryptoeprint:2018:227} recently proved a lower bound
showing that any oblivious algorithm for stable compaction or
merging in the balls-and-bins model must incur at least
$\Omega(N \log N)$ bandwidth. The balls-and-bins model
characterizes a wide class of natural algorithms where each
element is treated as an atomic ``ball'' with a numeric label;
the algorithm may perform arbitrary boolean computation on the
labels, but is only allowed to move the balls around and not
compute on their values. Our scheme works in the balls-and-bins
model, and thus shows \emph{for the first time} that the
multi-server setting can enable overcoming known lower bounds
in the single-server setting for oblivious algorithms.
Furthermore, for stable compaction and merging no previous
multi-server scheme was known that is asymptotically faster
than existing single-server algorithms, even in the weaker
setting of computational security. We note finally that our
protocols are asymptotically optimal since clearly any correct
algorithm has to read the entire array.

\ignore{
\kartik{communication bandwidth is worst case log$^2$}

\kartik{need to maintain position labels across depths and need
  to pass it upper level. + we need to merge levels too. Both of
  these operations require oblivious sort. however, linear
  time/linear communication sort is not possible with constant
  servers. hence, we maintain shares of the data such that at all
  points in time, the data is permuted, but some notion of
  sortedness is maintained i.e., some additional data is
  maintained on another server to sort it again. Once permuted,
  we ensure that we access the data only once.
  We use the
  additional information to sort it back. Once the data is
  sorted, we leverage new linear time/linear communication
  oblivious merge and oblivious compact
  algorithms.}

\kartik{can we deamortize?}
\kartik{we do not have expected bounds, but deterministic
  bounds.}

\kartik{Can we prove that it is not possible to sort in sub-n log
  n server computation / communication assuming constant servers
  allowing poly storage
  and a client allowing constant storage.}

\kartik{client need not do any work}
}

\subsection{Technical Roadmap}
\label{sec:roadmap}

Oblivious sorting is 
an essential building block in hierarchical ORAM schemes.
At a high level, the key idea is to replace 
oblivious sorting, 
which costs $O(n\log n)$ time on 
an array of length $n$, 
with cheaper, linear-time operations.
Indeed, this was also the idea of 
Lu and Ostrovsky~\cite{multiserveroram}, but they apply
it to a computationally secure hierarchical ORAM. 
Earlier single-server  
ORAM schemes are built from logarithmically many cuckoo hash tables 
of doubling size. Every time 
a memory request has been served, one needs to merge
multiple stale cuckoo hash tables into a  
newly constructed 
cuckoo hash table --- this was previously accomplished by oblivious sorting~\cite{oram03,oram09,oblivhash}. 
Lu and Ostrovsky show
how to avoid cuckoo hashing, 
by having one {\it permutation 
server} permute the data 
in linear time, 
and by having a separate {\it storage 
server}, that is unaware of the permutation, 
construct a cuckoo hash table from the permuted array 
in linear time (with the client's help). 
 Unfortunately, Lu and Ostrovsky's technique 
fails for the perfect security context due to its
intimate reliance on pseudorandom functions (PRFs)
and cuckoo hashing --- the former introduces computational assumptions
and the latter leads to statistical failures (albeit with negligible
probability).

We are, however, inspired by Lu and Ostrovsky's
permutation-storage-separation paradigm 
(and a similar approach that was described independently 
by Stefanov and Shi~\cite{multicloudoram}).
The key concept here is to have one permutation-server
that permutes the data; and 
have operations and accesses be performed by
a separate storage server that is unaware of the permutation applied.
One natural question is whether we can  
apply this technique to  
directly construct a linear-time multi-server 
oblivious sorting algorithm --- unfortunately 
we are not aware of any way to achieve this. 
Chan et al.~\cite{cryptoeprint:2017:914}
and Tople et al.~\cite{cryptoeprint:2017:885}
 show that assuming the data is already randomly permuted (and the
permutation hidden),
one can simply apply any 
comparison-based sorting algorithm and it would retain obliviousness.
Unfortunately, it is well-known that comparison-based sorting 
must incur $\Omega(n \log n)$ time, and 
this observation does not extend to non-comparison-based sorting 
techniques since in general RAM computations on 
numeric keys can leak information through access patterns.

\paragraph{New techniques at a glance.}
We propose two novel techniques 
that allow us to achieve the stated results, both of which
rely on the permutation-storage-separation paradigm:
\begin{itemize}[leftmargin=5mm]
\item 
First, we observe that with multiple servers, 
we can adapt 
the single-server perfect ORAM scheme by
Chan et al.~\cite{perfectopram}
into a new variant such that 
reshuffling operations 
which was realized with oblivious sorting in Chan et al.~\cite{perfectopram}
can now be expressed entirely  
with merging and stable compaction operations without oblivious sorting.

\ignore{
Note in a single-server setting, this does not
help us improve asymptotical performance
in the balls-and-bins model, since oblivious merging is just
as expensive as oblivious sorting~\cite{}.
\elaine{cite Lin, Xie, Shi} 
}
\item 
Despite the known  
lower bounds 
in the single-server setting~\cite{cryptoeprint:2018:227}, 
we show that with multiple servers, 
we can indeed achieve linear-time oblivious merging
and oblivious stable compaction.
As many have  
observed earlier~\cite{circuitopram,Goodrich11,cryptoeprint:2017:914,oblivhash}
merging and compaction are also important building
blocks in the design of many oblivious algorithms --- we thus
believe that our new building blocks are of independent interest.
\end{itemize}

\subsubsection{Stable Compaction and Merging.}
We first explain the intuition behind our stable compaction algorithm.
For simplicity, for the time being we will consider only 2 servers and 
assume perfectly secure encryption for free 
(this assumption can later be removed with secret 
sharing and by introducing one additional server).
Imagine that we start out with an array of length $N$ that is encrypted
and resides on one server. 
The elements in the array are either real or dummy, and we would like
to move all dummy elements to the end of the array 
while preserving the order of the real elements as 
they appear in the original array. 
For security, we would like that any single server's view
in the protocol leaks no information about the array's contents.

\paragraph{Strawman scheme.}
An extremely simple strawman scheme is the following:
the client makes a scan of the input array 
on one server; whenever it encounters 
a real element, it re-encrypts it and writes 
it to the other server by 
appending it to the end of the output array (initially the output array
is empty). 
When the entire input array 
has been consumed, the client 
pads the output array 
with an appropriate number of (encrypted) dummy elements.

At first sight, 
this algorithms seems to preserve security: each server basically observes
a linear scan of either the input or the output array; and the perfectly
secure encryption hides array contents.
However, upon careful examination, the second server
can observe the time steps
in which a write happened to the output array --- and this leaks
which elements are real in the original array.
Correspondingly, 
in our formal modeling later (Section~\ref{sec:defn}), 
each server cannot only
observe each message sent
and received by itself, but also the 
time steps in which these events occurred.

\paragraph{A second try.}
For simplicity we will describe our approach with server computation
and server-to-server communication --- but it is not hard
to modify the scheme such that servers are completely passive.
Roughly speaking, 
the idea is for the first server (called the {\it permutation server}) 
to randomly permute
all elements and store a permuted array
on the second server (called the {\it storage
server}), such that the 
permutation is hidden from the storage server.
Moreover, in this permuted array, we would like 
the elements to be tagged with pointers to form
two linked lists: a real linked list and a dummy linked list.
In both linked lists, the ordering of elements
respects that of the original array.
If such a permuted array encoding two linked lists can be constructed,  
the client 
can simply traverse the real
linked list first from the storage server, 
and then traverse the dummy linked list --- 
writing down each element it encounters
on the first server (we always 
assume re-encryption upon writes).
Since the storage server does not know the random permutation
and since every element is accessed exactly once, it observes completely
random access patterns; 
and thus it cannot gain any secret information.

The challenge remains 
as to how to tag each element with 
the position of the next element in the permuted array.
This can be achieved in the following manner:
the permutation server first 
creates a random permutation in linear time (e.g.,
by employing Fisher-Yates~\cite{donald1998art}), such
that each element in the input array is now tagged
with where it wants to be in the permuted array (henceforth
called the position label).
Now, the client makes a reverse scan 
of this input array. During this process, it remembers
the position labels of the last real element seen 
and of the last dummy element seen so far --- this takes
$O(1)$ client-side storage.
Whenever a real element is encountered, the client tags
it with the position label of the last real seen.
Similarly, whenever a dummy is encountered, the client
tags it with the position label of the last dummy seen.
Now,
the permutation server 
can permute the array based on the predetermined permutation 
(which can also be done in linear time). At this moment,
it sends the permuted, re-encrypted array to the storage server
and the linked list  can now be traversed from the storage server
to read real elements followed by dummy elements. 

It is not difficult to see that 
assuming that the encryption scheme is perfectly secure and every
write involves re-encrypting the data, 
then the above scheme achieves perfect security 
against any single semi-honest corrupt server,
and completes in linear time.
Later we will replace the 
perfectly secure encryption 
with secret sharing and this requires the introduction of one additional server.

\paragraph{Extending the idea for merging.}
We can extend the above idea to allow linear-time oblivious merging 
of two sorted arrays.
The idea is to prepare both arrays such that 
they are in permuted form on the storage server
and in a linked list format;
and now the client 
can traverse the 
two linked lists on the storage server, merging them 
in the process.
In each step of the merging, only one array is being consumed --- 
since the storage server does not know the permutation, 
it sees random accesses and cannot tell which array is being consumed.

\subsubsection{3-Server Perfectly Secure ORAM.}

We now explain the techniques for constructing a 3-server perfectly secure ORAM.
A client, with $O(1)$ blocks of local cache, 
stores $N$ blocks of data (secret-shared)
on the 3 servers, one of which might be semi-honest corrupt.
In every iteration, the client receives 
a memory request of the form $({\sf read}, \addr)$
or $({\sf write}, \addr, \data)$, and it 
completes this request 
by interacting with the servers.
We would like to 
achieve $O(\log^2 N)$ amortized bandwidth blowup per logical memory request. 

We start out from a state-of-the-art 
single-server perfectly-secure scheme by Chan et
al.~\cite{perfectopram} 
that achieves $O(\log^3 N)$ amortized bandwidth blowup per memory
request. \kartik{should explain [5]?}
Their scheme follows the hierarchical ORAM paradigm~\cite{oram00,oram10}
and meanwhile relies on a standard recursion technique
most commonly adopted by tree-based ORAMs~\cite{asiacrypt11}.
In their construction, there are logarithmically 
many hierarchical ORAMs (also called position-based ORAMs), where
the ORAM at depth $d$ (called the parent depth) stores position labels 
for the ORAM at depth $d+1$ (called the child depth); and finally,
the ORAM at the maximum depth $D = O(\log N)$
stores the real data blocks.\anti{add the address}

As it turns out, 
the most intricate part of Chan et al.~\cite{perfectopram}'s scheme
is the information passing between 
an ORAM at depth $d$ and its parent ORAM at depth $d-1$ \anti{and this is where shuffle is required }.
As Chan et al. describe it, 
all the logarithmically many ORAMs must perform coordinated
reshuffles upon every memory request: during a reshuffle,
the ORAM at depth $d$ must pass 
information back to the parent depth $d-1$.
More specifically, the depth-$d$ ORAM is aware of the updated
position 
labels for blocks that have been recently visited, 
and this must be passed back to the parent depth
to be recorded there. 

More abstractly and somewhat imprecisely, 
here is a critical building block in Chan et al.~\cite{perfectopram}:
suppose that the parent and the child each has an array of logical addresses 
and a position label for each address.
It is guaranteed by the ORAM construction that all addresses
the child has must appear in the parent's array.
Moreover, 
if some address appears in both the parent and child, then the child's version
is fresher.
Now, we would like to combine the information 
held by the parent and the child, retaining the freshest
copy of position label for every address\anti{typo}. 
Chan et al. then relied on oblivious
sorting to achieve this goal: 
if some address
is held by both the parent and child, they will appear adjacent to each
other in the sorted array; and thus in a single linear scan 
one can easily cross out all stale copies. 

To save a logarithmic factor, 
we must solve the above problem using only 
merging and compaction and not sorting. 
Notice that if 
both the parent's and the child's arrays
are already sorted according to the addresses, 
then the afore-mentioned information propagation from child to parent 
can be accomplished through merging rather than sorting (in the full
scheme we would also need stable compaction to remove dummy blocks 
in a timely fashion to avoid blowup of array sizes over time).
But how can we make sure that these 
arrays are sorted in the first place without oblivious sorting?
In particular, these arrays actually correspond to levels in a hierarchical
ORAM in Chan et al.~\cite{perfectopram}'s 
scheme, and all blocks in a level must appear in randomly permuted
order to allow safe (one-time) accesses --- this 
seems to contradict our desire for sortedness.
Fortunately, 
here we can rely again on the permutation-storage-separation paradigm
\anti{define the paradigm above}
--- for simplicity again we describe our approach for 2 servers assuming 
perfectly secure (re-)encryption upon every write.
The idea is the following:
although the storage server is holding each array (i.e., level)
in a randomly permuted order, 
the permutation server will remember an inverse permutation
such that when this permutation is applied to the storage server's copy,
sortedness is restored.
Thus whenever shuffling is needed,
the permutation server would first apply the inverse permutation to 
the storage server's copy to restore sortedness, and then
we could 
rely on merging (and compaction) to propagate information
between adjacent depths rather than sorting.


\subsection{Related Work}

The notion of Oblivious RAM (ORAM) was introduced by the seminal
work of Goldreich and Ostrovsky around three decades
ago~\cite{oram00,oram10}.
Their construction used a hierarchy of buffers of
exponentially increasing size, which was later known as the
hierarchical ORAM framework.
Their construction achieved an
amortized bandwidth blowup of $O(\log^3 N)$ and was secure
against a computationally bounded adversary.
Subsequently, several works have improved the bandwidth blowup
from $O(\log^3 N)$ to $O(\log^2 N / \log\log
N)$~\cite{oram03,oram09,oblivhash,circuitopram}
under the same adversarial model.
Ajtai~\cite{ajtaioram} was the first to consider
 the notion of a statistically
secure oblivious RAM that achieves $O(\log^3 N)$ bandwidth
blowup. This was followed by the statistically secure ORAM
construction by Shi et al.~\cite{asiacrypt11}, who introduced the
tree-based paradigm.
ORAM constructions in the tree-based paradigm have improved the
bandwidth blowup from $O(\log^3 N)$ to $O(\log^2
N)$~\cite{asiacrypt11,PathORAM,clp,RingORAM,CircuitORAM}.
 Though
the computational assumptions have been removed, the
statistically secure ORAMs still fail with a
failure probability that is negligibly small in the number of
data blocks stored in the ORAM.

\paragraph{Perfectly-secure ORAMs.} Perfectly-secure
ORAM was first studied by
Damg{\aa}rd et al.~\cite{oram05}. Perfect security requires
that a computationally unbounded server does not learn
anything other than the number of requests with
probability 1. This implies that the oblivious program's
memory access patterns should be {\it identically distributed} regardless
of the inputs to the program; and thus with probability 1,
no information can be leaked about the secret inputs to the program.
Damg{\aa}rd et al.~\cite{oram05}  achieve an \emph{expected}
$O(\log^3 N)$
simulation overhead and $O(\log N)$ space blowup relative to the
original
RAM program.
Raskin et al.~\cite{cryptoeprint:2018:268}  and Demertzis et
al.~\cite{cryptoeprint:2017:749} achieve a \emph{worst-case}
bandwidth blowup of
$O(\sqrt{N}\frac{\log
  N}{\log \log N})$ and $O(N^{1/3})$, respectively.
Chan et al.~\cite{perfectopram} improve upon Damg{\aa}rd et
al.'s result~\cite{oram05} by avoiding the
$O(\log N)$ blowup in space, and  by showing a construction that
is conceptually simpler. Our construction builds upon Chan et
al. and improves
the bandwidth blowup to \emph{worst-case} $O(\log^2 N)$ while
assuming three non-colluding servers.

\paragraph{Multi-server ORAMs.}
ORAMs  in this category assume multiple non-colluding servers
to improve bandwidth
blowup~\cite{multiserveroram,abraham2017asymptotically,hoang2017s,cryptoeprint:2018:005,DBLP:journals/corr/abs-1802-05145}.
A comparison of the relevant schemes is presented in
Table~\ref{tab:comparison}. Among these, the work that is
closely related to ours is by Lu and
Ostrovsky~\cite{multiserveroram} which achieves a bandwidth
blowup of $O(\log N)$ assuming two non-colluding servers. In
their scheme, each server performs permutations for data that
is stored by the other server. While their construction is
computationally secure, we achieve perfect security for access
patterns as well as the data itself. Moreover, our techniques
can be used to perform an oblivious tight stable compaction and
an oblivious merge operation in linear time; how to perform
these operations in linear time were not known even for the
computationally secure setting. On the other hand, our scheme
achieves an $O(\log^2 N)$ bandwidth blowup and uses three
servers. We remark that if we assume a perfectly secure
encryption scheme, our construction can achieve perfectly
secure access patterns using two servers. Abraham et
al.~\cite{abraham2017asymptotically}, Gordon et
al.~\cite{cryptoeprint:2018:005} and Kushilevitz and
Mour~\cite{DBLP:journals/corr/abs-1802-05145} construct
multi-server
ORAMs using PIR. Each of these constructions require the server
to perform computation for using PIR operations. While Abraham
et al.\cite{abraham2017asymptotically} achieve statistical
security for access patterns, other
work~\cite{cryptoeprint:2018:005,DBLP:journals/corr/abs-1802-05145}
is only computationally secure. While the work of Gordon et
al.\ achieves a bandwidth blowup of $O(\log N)$, they require linear-time
server computation. Abraham et al.\ and Kushilevitz and Mour,
on the other hand, are poly-logarithmic and logarithmic
respectively, both in computation and bandwidth blowup.
In comparison, our
construction achieves perfect security and requires a passive
server (i.e., a server that does not perform any computation)
at a bandwidth blowup of $O(\log^2 N)$.

\begin{table*}[t]
  \centering
  \caption{\textbf{Comparison with existing multi-server
      Oblivious RAM schemes for block size $\Omega(\log N)$.} All
      of the other schemes (including the statistically-secure
      schemes~\cite{abraham2017asymptotically}) require two
      servers
      but assume the
      existence of
      an unconditionally secure encryption scheme. With a similar
      assumption, our work would indeed need only two servers
      too.
	}
  \label{tab:comparison}
  \vspace{6pt}
  \begin{tabular}{r|ccc}
    \multirow{2}{*}{\textbf{Construction}} & \textbf{Bandwidth} &
                                                                  \textbf{Server} & \multirow{2}{*}{\textbf{Security}}  \\
                & \textbf{Blowup} & \textbf{Computation} &
    \\\toprule
    Lu-Ostrovsky~\cite{multiserveroram} & $O(\log N)$ & - &
                                                            Computational
    \\

    Gordon et al.~\cite{cryptoeprint:2018:005} & $O(\log N)$ &
                                                             $O(N)$
                                                                                  &
                                                                                    Computational
    \\
    Kushilevitz et al.~\cite{DBLP:journals/corr/abs-1802-05145} &
                                                                  $O(\log
                                                                  N
                                                                  \cdot 
                                                                  \omega(1))$
                                                                &
                                                                  $O(\log
                                                                  N \cdot
                                                                  \omega(1))$
                                                                                  &
                                                                                    Computational\\
Abraham et al.~\cite{abraham2017asymptotically} & $O(\log^2 N
                                                  \cdot \omega(1))$ &
                                                             $O(\log^2
                                                                      N \cdot
                                                                      \omega(1))$
                                                                                  &
                                                            Statistical
    \\
    Our work & $O(\log^2 N)$ & - & Perfect\\

    \bottomrule
  \end{tabular}
\end{table*}



\section{Definitions}
\label{sec:defn}
In this section, we revisit how to define multi-server ORAM schemes 
for the case of semi-honest corruptions.
Our definitions require that the adversary, controlling a subset
of semi-honest corrupt servers, learns no secret information 
during the execution of the ORAM protocol.
Specifically 
our adversary can observe all messages transmitted
to and from corrupt servers, the rounds in which they were transmitted, 
as well as communication patterns between honest parties (including
the client and honest servers).
Our definition generalizes existing works~\cite{abraham2017asymptotically}
where they assume free encryption of data contents (even when statistical
security is desired).

\subsection{Execution Model}
\label{sec:execmodel}


\paragraph{Protocol as a system of Interactive RAMs.}
We consider a protocol between multiple parties including 
a client, henceforth denoted by $\C$, 
and $k$ servers, denoted by $\S_0, \ldots, \S_{k-1}$, respectively.
The client and all servers are Random Access Machines (RAMs) that interact
with each other. Specifically, the client or each server 
has a CPU capable of computation and a memory
that supports reads and writes; the CPU interacts with the memory
to perform computation.
The atomic unit of operation for memory is called a {\it block}.
\ignore{
Throughout the paper, we assume that 
that {\it the client can store $O(1)$ number of memory blocks}. 
but the servers have unbounded storage (although our protocols
later only need linear storage from the servers). 
}
We assume that all RAMs can be {\it probabilistic}, i.e., they can read  
a random tape supplying a stream of random bits.

\paragraph{Communication and timing.}
We assume pairwise channels between all parties. 
There are two notions of time in our execution model,
CPU cycles and communication rounds.
Without loss of generality, 
henceforth we assume that it takes the same amount of time
compute each CPU instruction 
and to transmit each memory block over
the network to another party (since we can 
always take the maximum of the two).
Henceforth in this paper we often use the word {\it round} to denote
the time that has elapsed since the beginning of the protocol.

\ignore{
Whenever an honest party sends a message to a recipient,
the message will arrive at the recipient at the beginning of the next round. 
}

Although we define RAMs on the servers as being capable of
performing any arbitrary computation, all of our protocols
require the servers to be passive, i.e., the server RAMs only
perform read/write operations from the memory stored by it.

\subsection{Perfect Security under a Semi-Honest Adversary}\label{sec:IRpar}
We consider the client to be {\it trusted}. The adversary
can corrupt 
a subset of the servers 
(but it cannot corrupt the client) --- although our constructions later secure
against any individual corrupt server, we present definitions for the more
general case, i.e., 
when the adversary can control 
more than one corrupt server.

We consider 
a {\it semi-honest} adversary, i.e., the corrupt servers
still honestly follow the protocol; however, we would like to 
ensure that no undesired information will leak.
To formally define security, we need to first define
what the adversary can observe in a protocol's execution.

\paragraph{View of adversary $ \view^\adv$.}
Suppose that the adversary \algA controls a subset of 
the servers --- 
we abuse notation and 
use $\algA \subset [k]$ to denote the set 
of corrupt servers.
The view of the adversary, denoted by $\view^\adv$ in a random run of the protocol
consists of the following:
\begin{enumerate}[leftmargin=5mm]
\item 
{\it Corrupt parties' views,}
including 
1) corrupt parties' inputs,
2) all randomness consumed by corrupt parties, and 
3) an ordered sequence of all messages received by corrupt parties, including
which party the message is received from, as well as the {\it round} in which
each message is received. 
We assume that these messages are ordered by the round in which
they are received, and then by the party from which it is received.
\item 
{\it Honest communication pattern}:
when honest parties (including the client) exchange messages,
the adversary observes their communication pattern:
including which pairs of honest nodes exchange messages in which round.
\end{enumerate}
We stress that 
in our model only 
one block can be exchanged between every pair in a round --- thus
the above $\view^{\algA}$ definition  
effectively allows $\algA$ to see the total length of messages 
exchanged 
between honest parties.




\begin{remark}
We remark that this definition captures a notion of timing
patterns along
with  access patterns. For instance, suppose two servers store
two sorted lists that needs to be merged. The client performs a
regular merge operation to read from the two lists, reading the
heads of the lists in each round. In such a scenario, depending
on the rounds in which blocks are read from a server, an
adversary that corrupts that server can compute the
relative ordering of blocks between the two lists.
\end{remark}

\paragraph{Defining security in the ideal-real paradigm.}
Consider an ideal functionality 
$\algF$: upon receiving 
the input $\Inp_0$ from the client 
and inputs $\Inp_1, \ldots, \Inp_k$ from each of the $k$ servers respectively, 
and a random string $\rho$ sampled from some distribution,
$\algF$ computes 
$$(\Out_0, \Out_1, \ldots, \Out_k) 
:= \algF(\Inp_0, \Inp_1, \ldots, \Inp_k;\rho)$$ where 
$\Out_0$ is the client's output, and
$\Out_1, \ldots, \Out_k$ denote the $k$ servers' outputs respectively.  


\begin{definition}[Perfect security in the presence of a semi-honest adversary]\label{secdef}
We say that ``a protocol $\Pi$ perfectly securely realizes 
an ideal functionality 
$\algF$ 
in the presence of a semi-honest adversary corrupting $t$ servers'' 
iff for every adversary $\algA$ that controls up to $t$ corrupt servers, 
there exists 
a simulator $\Sim$ such that  
for every input vector $(\Inp_0, \Inp_1, \ldots, \Inp_k)$,
the following real- and ideal-world experiments 
output identical distributions: 

\begin{itemize}[leftmargin=5mm]
\item 
{\it Ideal-world experiment.}
Sample $\rho$ at random and compute 
$(\Out_0, \Out_1, \ldots, \Out_k) := \algF(\Inp_0, \Inp_1, \ldots, \Inp_k, \rho)$.
Output the following tuple
where we abuse notation and use $i \in \algA$ to denote the fact
 that $i$ is corrupt:
$$\Sim(\{\Inp_i, 
\Out_i\}_{i\in\adv}), \ \ 
\Out_0,  \{\Out_i\}_{i\not\in\adv}
$$

\item 
{\it Real-world experiment.}
Execute the (possibly randomized) real-world protocol, 
and let $\Out_0, \Out_1, \ldots, \Out_k$ be the outcome of the client
and each of the $k$ servers respectively.
Let $\view^\adv$ denote the view 
of the adversary \algA in this run.
Now, output the following tuple:
\[
\view^\adv, \ \Out_0,  \{\Out_i\}_{i\not\in\adv}
\]



\end{itemize}
\end{definition}


Note that throughout the paper, we will define various building 
blocks that realize different ideal functionalities. 
The security of all building blocks can be defined
in a unified approach with this paradigm.
When we compose these building blocks to construct our full protocol,
we can prove perfect security of the full protocol
in a composable manner. By modularly proving  
the security of each building block, we can 
now think of each building block 
as interacting with an ideal functionality. 
This enables us to 
prove the security of the full protocol in the ideal world
assuming the existence of these ideal functionalities.

\paragraph{Active-server vs. passive-server protocols.}
In active-server 
protocols, servers can perform arbitrary 
computation and send messages to each other.
In passive-server protocols, the servers act as passive memory and can
only answer memory read and write  
requests from the client; furthermore, 
there is only client-server communication and servers do not 
communicate with each other.
Obviously passive-server
schemes are more general --- 
in fact, all schemes in this paper are in the passive-server paradigm. 
We stress that all of our security definitions
apply to both passive-server and  
active-server schemes.


\subsection{Definition of $k$-Server Oblivious RAM}

\paragraph{Ideal logical memory.}
The ideal logical memory is defined in the most natural way.
There is a memory array consisting of $N$ blocks 
where each block is $\Omega(\log N)$ bits long, and each block
is identified by its unique {\it address} which takes value
in the range $\{0, 1, \ldots, N-1\}$.

Initially all blocks are set to 0.
Upon receiving 
$({\tt read}, \addr)$, the value of the block 
residing at address $\addr$
is returned.
Upon receiving 
$({\tt write}, \addr, \data)$, the block at address 
$\addr$ is overwritten with the data value $\data$, and its old value
(before being rewritten) is returned.

\paragraph{$k$-server ORAM.}
A $k$-server Oblivious RAM (ORAM) is a protocol between 
a client $\C$ 
and $k$ servers $\S_1, \ldots, \S_k$ 
which realizes an ideal logical memory.
The execution of this protocol 
proceeds in a sequence of iterations:
in each interaction, the client $\C$ receives a logical memory request
of the form
$({\tt read}, \addr)$
or $({\tt write}, \addr, \data)$.
It then engages in some (possibly randomized) 
protocol with the servers, at the end of which
it produces some output thus completing the current iteration.

We require perfect correctness and perfect security as defined below. 
We refer to a sequence of logical memory requests as a 
{\it request sequence} for short.
\begin{itemize}[leftmargin=5mm]
\item  {\it Perfect correctness.}
For any request sequence, 
with probability $1$, all of the client's outputs 
must be correct. In other words, we require that with probability 1, 
all of the client's outputs must match 
what the an ideal logical memory would have output for 
the same request sequence. 

\item  {\it Perfect security under a semi-honest adversary.} 
We say that a $k$-server ORAM scheme satisfies
perfect security w.r.t. a semi-honest adversary 
corrupting $t$ servers, 
iff 
for every \algA that controls up to $t$ servers, 
and for every two request sequences 
${\bf R}_0$
and ${\bf R}_1$ of equal length, 
the views $\view^\adv({\bf R}_0)$
and
$\view^\adv({\bf R}_1)$ are identically distributed, where 
$\view^\adv({\bf R})$
denotes the
view of \algA 
(as defined earlier in Section~\ref{sec:IRpar})
under the request sequence ${\bf R}$.

\ignore{
For $\adv \subseteq \{1, \cdots, k\}$ of at most $t \in \N$, let $\view^\adv({\bf R})$ be the view of the adversary as defined in Section \ref{sec:IRpar} given the request sequence ${\bf R}$. Note that in this case the view of the adversary also includes the length of each request sequence. We require that 
for any request sequences ${\bf R}_0$
and ${\bf R}_1$ of equal length, 
the views $\view^\adv({\bf R}_0)$
and 
$\view^\adv({\bf R}_1)$ are identically distributed.
}
\end{itemize}

\ignore{
So far, we have considered stateless (i.e., one-shot) ideal functionalities.
We would like to define a $k$-server ORAM 
as a protocol that securely realizes a logical memory 
abstraction henceforth denoted $\algF_{\rm mem}$. 
Such a logical memory abstraction is a {\it stateful} functionality
(also called a {\it reactive} functionality in the standard cryptography
literature~\cite{uc,Canetti2000}):
\begin{itemize}[leftmargin=5mm]
\item[]
$\mcal{F}_{\rm mem}$
can receive two types of instructions of the form $({\tt read}, \addr)$
or $({\tt write}, \addr, \data)$.
A {\tt read} instruction 
instructs $\algF_{\rm mem}$ to read and return the data residing
at logical address $\addr$.
A ${\tt write}$ instruction instructs  
$\algF_{\rm mem}$ to overwrite logical address $\addr$ with $\data$,
and the old contents of address $\addr$ 
are returned.
We assume that the logical memory is 
initialized as a all-$0$-vector.
\end{itemize}

\elaine{now, define realize for stateful}
}

Since we require perfect security (and our scheme is based on
information-theoretic secret-sharing), our notion
resists adaptive corruptions and is composable.
\ignore{
\begin{remark}
\elaine{say 
in all our defns,
since we are perfect sec, we resist adaptive corruption automatically}
\end{remark}
}

\subsection{Resource Assumptions and Cost Metrics}

We assume that the client can store
$O(1)$ blocks while the servers can store $O(N)$ blocks.
We will use the metric {\it bandwidth blowup} to  characterize
the performance of our protocols.
Bandwidth blowup is the (amortized) number of blocks queried in the ORAM
simulation to query a single virtual block.
We also note that since the servers do not perform any
computation, and the client always performs an $O(1)$ computation
on its $O(1)$ storage, an $O(X)$ bandwidth blowup also
corresponds to an $O(X)$ \emph{runtime} for our protocol.

\section{Core Building Blocks: Definitions and Constructions}

Imagine that there are three servers denoted $\S_0$, $\S_1$,
and $\S_2$, and a client denoted $\C$. We use~$\S_b, b \in
\zthree$ to refer to a specific server. Arithmetic performed on
the subscript $b$ is done modulo~3.

\subsection{Useful Definitions}
Let $\Arr$ denote a list of blocks where each block
is either a real block containing a payload string and a logical address;
or a dummy block denoted $\bot$.
We define {\it sorted} and
{\it semi-sorted} as follows:
\begin{itemize}[leftmargin=5mm]
\item
{\it Sorted:} $\Arr$ is said to be sorted \emph{iff}
all real blocks appear before dummy ones;
and all the real blocks appear in increasing order of
their logical addresses. If multiple blocks have the same logical address,
their relative order can be arbitrary.

\item
{\it Semi-sorted:}
$\Arr$ is said to be semi-sorted
\emph{iff} all the real blocks appear in increasing order of
their logical addresses, and ties may be broken arbitrarily.
However, the real blocks are allowed to be interspersed by dummy blocks.
\end{itemize}

%

\paragraph{Array Notation.} We assume each location of an array $\Arr$
stores a \emph{block} which is a bit-string of length~$B$.
Given two arrays $\Arr_1$ and $\Arr_2$,
we use $\Arr_1 \oplus \Arr_2$ to denote the resulting array after
performing bitwise-XOR on the corresponding elements at each index of the two arrays;
if the two arrays are of different lengths, we assume the shorter array is appended with
a sufficient number of zero elements.

\paragraph{Permutation Notation.} When a permutation $\pi: [n] \ra [n]$
is applied to an array $\Arr$ indexed by~$[n]$
to produce $\pi(\Arr)$, we mean
the element currently at location~$i$ will be moved to location~$\pi(i)$.
When we compose permutations, $\pi \circ \sigma$ means that
$\pi$ is applied \emph{before} $\sigma$. We use \iperm{}
to denote the identity permutation.

\paragraph{Layout.}
A layout is a way to store some data $\Arr$ on three servers
such that the data can be recovered by combining information on the three servers.
Recall that the client has only $O(1)$ blocks of space, and our protocol
does not require that the client stores
any persistent data.

Whenever some data $\Arr$ is stored on a server, informally speaking,
we need to ensure two things: 1) The server does not learn the
data $\Arr$ itself, and
2) The server does not learn \emph{which} index i of the data is
accessed.
In order to ensure the prior, we XOR secret-share the data $\Arr
:= \Arr_0 \oplus \Arr_1 \oplus \Arr_2$
between
three servers $\S_b, b \in \zthree$ such that $\S_b$ stores $\Arr_b$.
For a server to not learn \emph{which}
index $i$ in $\Arr$ is accessed, we ensure that the data is
permuted, and the access happens to the permuted data. If the
data is accessed on the same server that permutes the data, then
the index $i$ will still be revealed. Thus, for each share
$\Arr_{b}$, we ensure that one server permutes it and we access
it from another server, i.e.,
we have two types of servers:

\begin{itemize}[leftmargin=5mm]
\item
Each server $\S_b$ acts as a {\it storage server} for the $b$-th share, and thus
it knows $\Arr_b$.
\item
Each server $\S_b$ also acts as the {\it permutation server}
for the $(b+1)$-th share, and thus
it also knows $\Arr_{b+1}$ as well as $\pi_{b+1}$.
\end{itemize}

Throughout the paper, a layout is of the following form
\[
\text{3-server layout}: \quad \{\pi_b, \Arr_b\}_{b \in \zthree}
\] 
where $\Arr_b$ and $(\pi_{b+1}, \Arr_{b+1})$ are stored by server
$\S_b$. As mentioned, $\S_b$ not only knows its own share ($\Arr_b$) but
also the permutation and share of the next server $(\pi_{b+1}, \Arr_{b+1})$.




Specifically, $\Arr_0, \Arr_1, \Arr_2$ denote lists
of blocks of equal length: we denote $n = |\Arr_0| = |\Arr_1| = |\Arr_2|$.
Further, $\pi_{b+1}: [n] \rightarrow [n]$ is a permutation stored by
server $\S_{b}$ for the list $\Arr_{b+1}$.
Unless there is ambiguity, we use $\oplus_b$ to mean applying $\oplus_{b \in \zthree}$
to three underlying arrays.

The above layout is supposed to store
the array that can be recovered by:

$$\oplus_{b} \pi_{b}^{-1}(\Arr_b).$$



Henceforth, given a layout $\{\pi_b, \Arr_b\}_{b \in \zthree}$,
we say that the layout is {\it sorted}
(or semi-sorted) iff
$\oplus_{b} \pi_{b}^{-1}(\Arr_b)$ is sorted (or
semi-sorted).

\noindent \emph{Special Case.}
Sometimes the blocks secret-shared among $\S_0$, $\S_1$, $\S_2$
may be unpermuted, i.e.,
for each $b \in \zthree$,
$\pi_b$ is the identity permutation~$\iperm$.
In this case,
the layout is
\[
\text{Unpermuted layout}: \quad \{\iperm, \Arr_b\}_{b \in \zthree}
\]

For brevity, the unpermuted layout $ \{\iperm, \Arr_b\}_{b \in
  \zthree}$ is also denoted by the abstract array $\Arr$.

\begin{definition}[Secret Write]
\label{defn:secret_write}
An \emph{abstract}
array $\Arr$
corresponds to some unpermuted layout
$\{\iperm, \Arr_b\}_{b \in \zthree}$.
We say that the client \emph{secretly writes}
a value $\mathsf{B}$ to
the array $\Arr$ at index~$i$,
when it does the following:


\begin{compactitem}
\item Sample random values $\mathsf{B}_0$ and $\mathsf{B}_1$ independently,
and compute $\mathsf{B}_2 := \mathsf{B} \oplus \mathsf{B}_0 \oplus \mathsf{B}_1$.

\item For each $b \in \zthree$, the client writes $\Arr_b[i] := \mathsf{B}_b$ on server $\S_b$
(and $\S_{b-1}$).
\end{compactitem}

\end{definition}

\begin{definition}[Reconstruct]
Given some layout $\{\pi_b, \Arr_b\}_{b \in \zthree}$,
the client \emph{reconstructs} a value from using tuple $(i_0, i_1, i_2)$ of indices,
when it does the following:

\begin{compactitem}
\item For each $b \in \zthree$, the client reads $\Arr_b[i_b]$
from server $\S_b$. (It is important that the client reads $\Arr_b$ from $\S_b$,
even though $\Arr_b$ is stored in both $S_b$ and $S_{b-1}$.)
\item The reconstructed value is $\oplus_b \Arr_b[i_b]$.
\end{compactitem}
\end{definition}

\paragraph{Protocol Notation.}
All protocols are denoted as {\sf out} $\leftarrow$ {\sf
Prot}({\sf sin}, {\sf cin}). Here, {\sf sin} and {\sf cin} are
respectively server and client inputs to the protocol {\sf
  Prot}. Except for in an ORAM {\sf Lookup}, all the outputs {\sf out} are sent
to the server.



\subsection{Permute and Unpermute}
\label{sec:perm-unperm}
\subsubsection{Non-oblivious random permutation.}
\label{sec:fisher-yates}
Fisher and Yates~\cite{donald1998art} show how to generate a uniformly random
permutation $\pi: [n] \rightarrow [n]$ in $O(n)$ time steps. This
implies that the client can write a random permutation on a
server with $O(n)$ bandwidth.
The permutation is non-oblivious, i.e., the
server \emph{does} learn the permutation generated.

\subsubsection{Definition of ${\sf Permute}$.}
${\sf Permute}$
is a protocol that realizes an ideal functionality
$\mcal{F}_{\rm perm}$ as defined below.
Intuitively, this functionality takes some unpermuted input
layout
(i.e., unpermuted secret-shared inputs)
and three additional permutations $\pi_{b+1}$ from the three
permutation servers $\S_b$.
The functionality produces an output such that the three shares
are secret-shared again, and the share received by
storage server $\S_{b+1}$ is
permuted
using $\pi_{b+1}$.
Secret-sharing the data again before applying the new permutations
ensures that a storage server $\S_{b+1}$ does not learn the
permutation $\pi_{b+1}$
applied to its share.

\begin{itemize}[leftmargin=5mm]

\item $\{\pi_b, \Arr'_b\}_{b \in \zthree}\leftarrow{\sf Permute \big((\{\iperm, \Arr_b\}_{b \in \zthree},\{\pi_{b}\}_{b \in \zthree}),\bot\big)}$:

\begin{itemize}[leftmargin=5mm]
\item
{\it Input}: Let $\{\iperm, \Arr_b\}_{b \in \zthree}$
 be the unpermuted layout provided as input. (Recall that $\Arr_b$ and $\Arr_{b+1}$ are
stored in server $\S_b$.) 

Moreover, for each~$b \in \zthree$, $\S_b$ has an additional permutation $\pi_{b+1}$ as input
(which could be generated by the client for instance).

The arrays have the same length $|\Arr_0| = |\Arr_1| = |\Arr_2| = n$, for some $n$.
The client obtains $\bot$ as the input.

\item
{\it Ideal functionality $\mcal{F}_{\rm perm}$}:

Sample independently and uniformly random $\widehat{\Arr}_0, \widehat{\Arr}_1$ of length $n$.

Now, define $\widehat{\Arr}_2 := \widehat{\Arr}_0 \oplus \widehat{\Arr}_1 \oplus (\oplus_{b} \Arr_b)$,
i.e., $\oplus_b \widehat{\Arr}_b  = \oplus_{b} \Arr_b$.

For each $b \in \zthree$, define $\Arr_b' := \pi_b (\widehat{\Arr}_b)$.

The output layout is
$\{\pi_b, \Arr'_b\}_{b \in \zthree}$,
and the client's
output is $\bot$.


\end{itemize}
\end{itemize}

\subsubsection{Protocol ${\sf Permute}$.}

The implementation of $\mcal{F}_{\rm perm}$ proceeds as follows:
\begin{enumerate}[leftmargin=5mm]


\item \emph{Mask shares.}
  For each data block, the client first generates block
  ``masks'' that sum up to
  \emph{zero}, and then applies mask to $\Arr_{b+1}$ on server
  $\S_b$.
Specifically, the client does the following, for each $i \in [n]$:
\begin{compactitem}
\item Generate block ``masks'' that sum up to \emph{zero}, i.e., sample independent random blocks $\mathsf{B}^i_0$
and $\mathsf{B}^i_1$, and compute $\mathsf{B}^i_2 := \mathsf{B}^i_0 \oplus \mathsf{B}^i_1$.

\item Apply mask $\mathsf{B}^i_{b+1}$ to $\Arr_{b+1}[i]$ stored
  on server $\S_b$, i.e., for each $i \in [b]$, the client
writes $\widehat{\Arr}_{b+1}[i] \gets {\Arr}_{b+1}[i] \oplus \mathsf{B}^i_{b+1}$
on server $\S_b$.


\end{compactitem}

\item \emph{Permute share of $\S_{b+1}$ and send result to
    $\S_{b+1}$.}
The client uses $\pi_{b+1}$ to  permute a share on the
permutation server and then sends this permuted
share to the storage server, i.e.,
for each $b \in \zthree$, the client computes 
computes $\Arr_{b+1}' := \pi_{b+1}(\widehat{\Arr}_{b+1})$ on
server $\S_b$,
and sends the result $\Arr'_{b+1}$ to $\S_{b+1}$.
Each server $\S_b$ stores $\Arr'_b$ and $(\pi_{b+1}, \Arr'_{b+1})$; hence, the new layout $\{\pi_b, \Arr'_b\}_{b \in \zthree}$
is achieved.
\end{enumerate}

\begin{theorem}\label{thm:permute}
The ${\sf Permute}$ protocol perfectly securely realizes the
ideal functionality $\mcal{F}_{\rm perm}$ in the presence of a
semi-honest adversary corrupting a single server with $O(n)$ bandwidth.

\end{theorem}

\begin{proof}
By construction, the implementation of the protocol
applies the correct permutation on each server's array
and re-distributes the secret shares using fresh independent randomness.
Hence, the marginal distribution of the protocol's outputs
is exactly the same as that of the ideal functionality.

Fix some $b \in \zthree$
and consider the $\view^b$ of the corrupt server $\S_b$.
Since the leakage of $\mcal{F}_{\rm perm}$ is empty,
we will in fact show that
given the inputs and the outputs to server $\S_b$,
the $\view^b$ is totally determined and has no more randomness.
Hence, given $\Inp_b$ and conditioning on $\Out_b$,
the $\view^b$ is trivially independent of the outputs
of the client and other servers.

Then, given the inputs $\Inp_b$ and the outputs $\Out_b$ to
$\S_b$, a simulator simply returns 
the view of $\S_b$ uniquely determined by $\Inp_b$ and $\Out_b$.


The inputs to $\S_b$ 
are the arrays $\Arr_b$ and $\Arr_{b+1}$,
and the permutation $\pi_{b+1}$.
The outputs are the arrays $\Arr'_b$ and $\Arr'_{b+1}$,
and also the permutation $\pi_{b+1}$.
We next consider each part of $\view^b$.

\begin{enumerate}

\item \textbf{Communication Pattern.}  The communication pattern between the client and all the servers only depends on~$n$.

\item \textbf{Data Structure.}  We next analyze the intermediate data that is observed by
$\S_b$.  The arrays $\Arr'_b$ and $\Arr'_{b+1}$ are in $\S_b$'s outputs.
Hence, it suffices to consider the intermediate
array $\widehat{\Arr}_{b+1} = \pi_{b+1}^{-1}(\Arr'_b)$,
which is totally determined by the outputs.
%
\end{enumerate}

We have shown that the $\view^b$
is actually a deterministic function of the inputs and the outputs of $\S_b$,
as required.
\paragraph{Efficiency. }   Recall that it takes $O(1)$ time to process one block.
	From the construction, it is straightforward that
	linear scans are performed on the relevant arrays.
			
	In particular, the two steps -- masking shares and permuting the
  shares -- can be done with $O(n)$ bandwidth. 
	Moreover, the client can generate a permutation on the
  server with $O(n)$ bandwidth, when each block has $\Omega(\log n)$ bits, using the
  Fisher-Yates algorithm~\cite{donald1998art}.

\end{proof}

\ignore{
The proof can be found in the supplemental material, Section
\ref{ap:permute}.} \kartik{sec theorems are not tied with the defn.}

\subsubsection{Definition of $\sf Unpermute$.}
${\sf Unpermute}$ is a protocol that realizes an ideal functionality 
$\mcal{F}_{\rm unperm}$ 
as defined below.
Intuitively, this functionality 
reverses the effect of $\mcal{F}_{\rm perm}$.
It takes some permuted input layout,
and returns the corresponding unpermuted layout.
However, to avoid each server knowing its original permutation,
the contents of each entry 
needs to be secret-shared again.

\begin{itemize}[leftmargin=5mm]

\item $\{\iperm, \Arr'_b\}_{b \in \zthree}\leftarrow{\sf Unpermute (\{\pi_{b},\Arr_b\}_{b \in \zthree},\bot)}$:

\begin{itemize}[leftmargin=5mm]
\item 
{\it Input}: Let $\{\pi_b, \Arr_b\}_{b \in \zthree}$
 be the layout provided as input. (Recall that $\Arr_b$ and $(\pi_{b+1}, \Arr_{b+1})$ are
stored in server $\S_b$.) 


The arrays have the same length $|\Arr_0| = |\Arr_1| = |\Arr_2| = n$, for some $n$.
The client obtains $\bot$ as the input.

\item 
{\it Ideal functionality $\mcal{F}_{\rm unperm}$}:

Sample independently and uniformly random ${\Arr'}_0, {\Arr'}_1$ of length $n$.

Now, define $\Arr'_2 := \Arr'_0 \oplus \Arr'_1 \oplus (\oplus_{b} \pi_b^{-1}(\Arr_b))$,
i.e., $\oplus_b \Arr'_b  = \oplus_{b} \pi_b^{-1}(\Arr_b)$.


The output layout is 
$\{\iperm, \Arr'_b\}_{b \in \zthree}$,
and the client's 
output is $\bot$. 


\end{itemize}
\end{itemize}

\subsubsection{Protocol ${\sf Unpermute}$.}

The implementation of $\mcal{F}_{\rm unperm}$ proceeds as follows:
\begin{enumerate}[leftmargin=5mm]


\item \emph{Compute inverse permutations.} For each $b \in
  \zthree$, the client computes the inverse permutation
 $\widehat{\Arr}_{b+1} := \pi^{-1}_{b+1}(\Arr_{b+1})$ on server $\S_{b}$.

\item \emph{Mask shares.} For each data block, the client
  generates
block ``masks'' that sum up to \emph{zero} and then applies the
mask to $\Arr_{b+1}$ on server $\S_b$. 
Specifically, the client performs the following, for each $i \in
[n]$: 
\begin{compactitem}
\item Generate block ``masks'' that sum up to zero, i.e., Sample independent random blocks $\mathsf{B}^i_0$
and $\mathsf{B}^i_1$, and compute $\mathsf{B}^i_2 := \mathsf{B}^i_0 \oplus \mathsf{B}^i_1$.

\item Apply mask $\mathsf{B}^i_{b+1}$ to $\Arr_{b+1}[i]$ stored
  on server $\S_b$, i.e., for each $i \in [b]$, the client 
writes $\Arr'_{b+1}[i] \gets \widehat{\Arr}_{b+1}[i] \oplus \mathsf{B}^i_{b+1}$
on server $\S_b$.


\end{compactitem}

\item For each $b \in \zthree$, the server $\S_b$
sends $\Arr'_{b+1}$
to $\S_{b+1}$.

Hence, the new layout $\{\iperm, \Arr'_b\}_{b \in \zthree}$
is achieved.

\end{enumerate}

\ignore{
  \begin{lemma}[Perfectly Secure Implementation of $\mcal{F}_{\rm unperm}$]
\label{lemma:sec_unperm}
The above implementation of $\mcal{F}_{\rm unperm}$
is perfectly secure as per Definition~\ref{secdef}.
\end{lemma}
}

\begin{theorem}\label{thm:unpermute}
The ${\sf Unpermute}$ protocol perfectly securely realizes the
ideal functionality $\mcal{F}_{\rm unperm}$ in the presence of a
semi-honest adversary corrupting a single server with  $O(n)$
bandwidth blowup.

\end{theorem}

\begin{proof}
The proof is essentially the same as Theorem~\ref{thm:permute}.
It can be checked that
for each $b \in \zthree$,
the $\view^b$ is a deterministic function
of the inputs and the outputs of $\S_b$.
\end{proof}

\ignore{
\subsubsection{Definition of ${\sf Unpermute}$ and Protocol Description.}
Similar to {\sf Permute}, we also need a complementary {\sf
  Unpermute} protocol. Its definition and protocol
are described in the
supplemental material, Section~\ref{appendix:unpermute}.
}

\subsection{Stable Compaction}
\label{sec:stable-compact}
\subsubsection{Definition of ${\sf StableCompact}$.}
${\sf StableCompact}$
is a protocol that realizes an ideal functionality
$\mcal{F}_{\rm compact}$,
as defined below:
\begin{itemize}[leftmargin=5mm]
\item  $ \{\iperm, \Arr'_b\}_{b \in \zthree} \leftarrow{\sf
    StableCompact}(\{\iperm, \Arr_b\}_{b \in \zthree}, \bot)$:
\begin{itemize}[leftmargin=5mm]
\item {\it Input layout}:
A {\it semi-sorted}, unpermuted layout
denoted $\{\iperm, \Arr_b\}_{b \in \zthree}$.


\item
{\it Ideal functionality $\mcal{F}_{\rm compact}$}:
$\mcal{F}_{\rm compact}$
computes $\Arr^* := \Arr_0 \oplus \Arr_1 \oplus \Arr_2$; it then
moves all dummy blocks in $\Arr^*$ to the end of the array,
while keeping the relative order of real blocks unchanged.

Now, $\mcal{F}_{\rm compact}$
randomly samples $\Arr'_0, \Arr'_1$ of appropriate length
and computes $\Arr'_2$ such that
$\Arr^* = \Arr'_0 \oplus \Arr'_1 \oplus \Arr'_2$.
The output layout is a {\it sorted}, unpermuted layout
$\{\iperm, \Arr'_b\}_{b \in \zthree}$.

\end{itemize}
\end{itemize}

\subsubsection{${\sf StableCompact}$ Protocol.}

The input is a {\it semi-sorted},
unpermuted layout, and we would like to
turn it into a {\it sorted}, unpermuted layout obliviously.
The key idea is to permute each share of the list (stored
on the 3 servers respectively), such that the storage server
for each share does not know the permutation.
Now, the client accesses all real elements in a sorted order,
and then accesses all dummy elements, writing down the elements
in a secret-shared manner as the accesses are made.
We can achieve this if
each real or dummy element is tagged with a pointer to its next
element, and the pointer is in fact a $3$-tuple
that is also secret shared on the 3~servers --- each element in the 3-tuple
indicates where the next element is in one of the 3~permutations.

Therefore, the crux of the algorithm is to tag each (secret-shared) element with
a (secret-shared) position tuple, indicating where its next element is ---
this will effectively create two linked list structures (one for real
and one for dummy): each element in the linked
lists is secret shared in to 3~shares,
and each share resides on its storage server at an independent
random location.

The detailed protocol is as follows:

\begin{enumerate}[leftmargin=5mm]

\item \label{compact:step1}
First, each server ${\bf S}_b$ acts as the permutation server
for ${\bf S}_{b + 1}$.
Thus, the client generates a random permutation $\pi_{b+1}$ on
the permutation server
$\S_{b}$ using the Fisher-Yates algorithm described in
Section~\ref{sec:fisher-yates}.
Basically, for each index $i$ of the original list the client
writes down, on each ${\bf S}_b$, that its $(b+1)$-th
share (out of 3 shares), wants to be in position $\pi_{b+1}(i)$.


\item \label{compact:step2} 
  Next, the client makes a reverse scan of
  $({\sf T}_0, \pi_0), ({\sf T}_1, \pi_1),
  ({\sf T}_2, \pi_2)$ for $i = n$ down to $1$.
  The client can access $({\sf T}_{b+1}[i], \pi_{b+1}(i))$ by
  talking to ${\bf S}_{b}$. 
  In this reverse scan, the client always locally remembers
  the position tuple of the last real element encountered
  (henceforth denoted $\mathfrak{p}_{\rm real}$)
  and the position tuple of the last dummy element encountered
  (henceforth
  denoted $\mathfrak{p}_{\rm dummy})$.
  Thus, if $\Arr[k_{\rm real}]$ is the last seen real
  element, then the client remembers $\mathfrak{p}_{\rm real} =
  (\pi_b(k_{\rm real}): b \in \zthree)$. $\mathfrak{p}_{\rm
    dummy}$ is updated analogously.
  Initially,
  $\mathfrak{p}_{\rm real}$ and $\mathfrak{p}_{\rm dummy}$ are
  set to $\bot$.

  During this scan, whenever a real element $\Arr[i]$
  is
  encountered,
  the client secretly writes the link $\link[i] :=
  \mathfrak{p}_{\rm real}$, i.e., $\link[i]$ represents
  secret shares of the next pointers for the real element and
  $\link$ itself represents an abstract linked list of real
  elements.
  The links for dummy elements are updated analogously using
  $\mathfrak{p}_{\rm dummy}$.

  At the end of this reverse scan, the client remembers
  the position tuple for the first real of the linked list
  denoted $\mathfrak{p}^1_{\rm real}$
  and position tuple for the first dummy
  denoted $\mathfrak{p}^1_{\rm dummy}$.

\item \label{compact:step3}
Next, we call ${\sf Permute}$
inputting
1) the original layout --- but importantly, now
each element is tagged with a position tuple (that is also secret shared);
and 2)
the three permutations
chosen by each ${\bf S}_b$ (acting as the permutation server for
${\bf S}_{b+1}$).
Thus,
${\sf Permute}$ is applied
to the combined layout $\{\iperm, (\Arr_b, \link_b)\}_{b \in \zthree}$,
where $\S_b$ has input permutation $\pi_{b+1}$.
Let the output of {\sf Permute} be denoted by
$\{\pi_b, (\Arr'_b, \link'_b)\}_{b \in \zthree}$.

\item \label{compact:step4}
Finally, the client
traverses first the real linked list
(whose start position tuple is $\mathfrak{p}^1_{\rm real}$)
and then the dummy linked list
(whose start position tuple is $\mathfrak{p}^1_{\rm dummy}$).
During this traversal, the client secretly writes
each element encountered to produce the sorted and unpermuted
output layout.

More precisely, the client secretly writes an abstract array
$\Arr''$ element by element.
Start with $k \gets 0$ and $\mathfrak{p} \gets \mathfrak{p}^1_{\rm
  real}$.

The client reconstructs element $\mathsf{B} := \oplus
\Arr'_b[\mathfrak{p}_b]$
and the next pointer of the linked list $\mathsf{next} := \oplus
\link'_b[\mathfrak{p}_b]$;
the client secretly writes to the abstract array $\Arr''[k] :=
\mathsf{B}$. 

Then, it updates $k \gets k + 1$ and $\mathfrak{p} \gets \mathsf{next}$,
and continues to the next element;
if the end of the real list is reached, then it sets $\mathfrak{p} \gets
\mathfrak{p}^1_{\rm dummy}$.
This continues until the whole (abstract) $\Arr''$ is secretly
written to the three servers.
	
\item \label{compact:step5}	The new layout $\{\iperm, \Arr''_b\}_{b \in \zthree}$ is constructed.
\end{enumerate}

\begin{theorem}\label{thm:compaction}
The ${\sf StableCompact}$ protocol perfectly securely realizes
the ideal functionality $\mcal{F}_{\rm compact}$
in the presence of a
semi-honest adversary corrupting a single server with $O(n)$ bandwidth.

\end{theorem}

\begin{proof}
By construction, the protocol correctly removes dummy elements and preserves
the original order of real elements, where the secret shares are re-distributed
using independent randomness.  Hence, the marginal distribution on the outputs
is the same for both the protocol and the ideal functionality.

We fix the inputs of all servers, and some $b \in \zthree$. 
The goal is to show that (1) the $\view^b$ follows a distribution
that is totally determined by the inputs $\mathbf{I}_b$ and the outputs $\mathbf{O}_b$ of the corrupt $\S_b$;
(2) conditioning on $\Out_b$,
$\view^b$ is independent of the outputs of the client and all other servers.

The second part is easy, because the inputs are fixed.
Hence, conditioning on $\Out_b$ (which includes $\Arr''_b$ and $\Arr''_{b+1}$),  
$\Arr''_{b+2}$ has no more randomness and totally determines the outputs of other servers.
%
%

To prove the first part, our strategy is to decompose $\view^b$ into a list components,
and show that fixing $\mathbf{I}_b$ and  conditioning on  $\mathbf{O}_b$ and a prefix 
of the components, the distribution of the next component can be determined.
Hence, this also gives the definition of a simulator.

First, observe that in the last step, the client
re-distributes the shares, and gives
output~$\mathbf{O}_b$ to $\S_b$;
moreover, the shares of $\Arr''$ are generated with fresh independent
randomness.  Hence, the distribution of the part of $\view^b$ excluding $\mathbf{O}_b$ is independent of $\mathbf{O}_b$.
We consider the components of $\view^b$ in the following order,
which is also how a simulator generates a view after seeing $\Inp_b$.

\begin{enumerate}

\item \emph{Communication Pattern.}
Observe that from the description of the algorithm,
the communication pattern between the client and the servers
depends only on the length~$n$ of the input array.

\item \emph{Random permutation $\pi_{b+1}$.}  This is independently generated
using fresh randomness.

\item \emph{Link $\link$ creation.}  The (abstract) array $\link$
is created by reverse linear scan.  The shares $\link_b$ and $\link_{b+1}$
received by $\S_b$ are generated by fresh independent randomness.\anti{again need to refer to the privacy of the sharing for the linked lists}

\item \emph{$\mathsf{Permute}$ subroutine.}  By the correctness
of the $\mathsf{Permute}$, the shares of the outputs $(\Arr', \link')$ 
received by $\S_b$ follow an independent and uniform random distribution.
By the perfect security of 
$\mathsf{Permute}$, the component of $\view^b$ due to
$\mathsf{Permute}$ depends only on the inputs (which include the shares
of $\Arr$ and $\link$) and the outputs of the subroutine.

\item \emph{List traversal.} Since $\S_b$ does not know $\pi_b$ (which is generated
using independent randomness by $\S_{b-1}$),
from $\S_b$'s point of view,
its array $\Arr'_b$ is traversed in an independent and uniform random order.
\end{enumerate}

Therefore, we have described a simulator procedure that
samples the $\view^b$ step-by-step, given $\mathbf{I}_b$
and $\mathbf{O}_b$.

\paragraph{Efficiency. }

  Each of the steps in the protocol can be executed with a
  bandwidth of
  $O(n)$. 
	Step~\ref{compact:step1} can be performed using
  Fisher-Yates shuffle algorithm. In steps~\ref{compact:step2} and
  \ref{compact:step4}, the client linearly scans the abstract
  lists $\Arr, \Arr''$ and the links $\link, \link'$. Accessing
  each array costs $O(n)$ bandwidth.
	Finally,
  step~\ref{compact:step3} invokes ${\sf permute}$, which
  requires $O(n)$ bandwidth
  (Theorem~\ref{thm:permute}).

\end{proof}

\ignore{
The proof can be found in the supplemental material, Section
\ref{ap:compaction}. 
}
\kartik{correspond to definition 2.1}

\subsection{Merging}

\subsubsection{Definition of ${\sf Merge}$.}
${\sf Merge}$
is a protocol that realizes an ideal functionality
$\mcal{F}_{\rm merge}$
as defined below:

\begin{itemize}[leftmargin=5mm]
\item $\{\iperm, \U''_b\}_{b \in \zthree}\leftarrow{\sf{Merge}\big(\{\iperm, (\Arr_b, \Arr'_b)\}_{b \in \zthree},\bot\big)}$:
\begin{itemize}[leftmargin=5mm]
\item
{\it Input layout}:
Two {\it semi-sorted}, unpermuted layouts
denoted $\{\iperm, \Arr_b\}_{b \in \zthree}$
and $\{\iperm, \Arr'_b\}_{b \in \zthree}$ denoting abstract
lists $\Arr$ and $\Arr'$,
where all the arrays have the same length~$n$.

\item
{\it Ideal functionality $\mcal{F}_{\rm merge}$}:
First, $\mcal{F}_{\rm merge}$
merges the two lists $\Arr_0 \oplus \Arr_1 \oplus \Arr_2$ and
$\Arr'_0 \oplus \Arr'_1 \oplus \Arr'_2$,
such that the resulting array is sorted with all dummy blocks
at the end. Let $\U''$ be this merged result.
Now, $\mcal{F}_{\rm merge}$
randomly samples $\U''_0$ and $\U''_1$
independently of appropriate length
and computes $\U''_2$ such that
$\U'' = \U''_0 \oplus \U''_1 \oplus \U''_2$.
The output layout is a {\it sorted}, unpermuted layout
$\{\iperm, \U''_b\}_{b \in \zthree}$.

\end{itemize}
\end{itemize}

\subsubsection{${\sf Merge}$ Protocol.}
\ignore{
A first attempt to merge two lists  is the following. For
simplicity, let us assume that the lists are encrypted using a
perfectly-secure encryption scheme. Since the lists are stored on
different servers, accesses
made
to one server is not visible to the other server. Thus, one can
perform a \emph{non-oblivious} merge operation by reading heads of the
list from the two servers and then writing the result to a third
server. At the end, the client pads the result to a fixed length
with dummy blocks.
From each servers perspective, the client only sequentially reads
the lists and the access pattern is individually oblivious.
Moreover, from
the third server's perspective, a data block is
written to it at every time unit independent of the relative
ordering of blocks between the two lists. However, similar to the
na\"ive solution described in the stable compact protocol,
the time interval between accesses of blocks can reveal the
relative ordering of the two lists. Our solution uses three
servers, storing both lists secret-shared and appropriately permuting
them before accessing the lists.
}

The protocol receives as input, two semi-sorted, unpermuted
layouts and produces a merged, sorted, unpermuted layout
as the output. The key idea is to permute the concatenation of
the two semi-sorted
inputs such that the storage servers do not know
the permutation.
Now, the client accesses real elements in both lists
in the
sorted order using the storage servers to produce a merged
output. Given that a
\emph{concatenation of the lists} is permuted together, elements
from \emph{which} list is accessed is not revealed during the
merge operation,
thereby allowing us to merge
the two lists obliviously.
In order to access the two lists in a sorted order, the client
creates a linked list of real and dummy elements using the
permutation servers, similar to the {\sf
StableCompact} protocol in Section~\ref{sec:stable-compact}.

The detailed protocol works as follows:
\begin{enumerate}[leftmargin=5mm]

\item First, the client concatenates the two abstract lists
  $\Arr$ and $\Arr'$ to obtain an abstract list $\U$ of size
  $2n$, i.e.,
  we interpret $\U_b$ as the concatenation of
  $\Arr_b$ and $\Arr'_b$ for each $b \in \zthree$.
  Specifically, $\U_b[0,n-1]$ corresponds to $\Arr_b$
  and $\U_b[n,2n-1]$ corresponds to $\Arr'_b$.

\item Now, each server $\S_b$ acts as the permutation server for
  $\S_{b+1}$.
  The client generates a random permutation $\pi_{b+1}: [2n] \ra
  [2n]$ on server
  $\S_{b+1}$ using the Fisher-Yates algorithm described in
  Section~\ref{sec:fisher-yates}. $\pi_{b+1}(i)$ represents the
  position of
  the $(b+1)$-th share and is stored on server $\S_b$.

%
%

\item The client now performs a reverse scan of $(\U_0, \pi_0),
  (\U_1, \pi_1), (\U_2, \pi_2)$ for $i = n$ down to $1$.
  During this reverse scan, the client always locally remembers
  the position tuples of the last real element and last dummy
  element encountered for both the lists. Let them be denoted by
  $\mathfrak{p}_{\rm real}$, $\mathfrak{p}'_{\rm real}$,
  $\mathfrak{p}_{\rm dummy}$, and $\mathfrak{p}'_{\rm dummy}$.
  Thus, if $\U[k_{\rm real}]$ is the last seen real element from
  the first list, the client remembers $\mathfrak{p}_{\rm real} =
  (\pi_b(k_{\rm real}) : b \in \zthree)$. The other position
  tuples are updated analogously. Each of these tuples are
  initially set to $\bot$.

  During the reverse scan, the client maintains an abstract linked list
  $\link$ in the following manner.
  When $\U[i]$ is processed, if it is a real
  element from the first list, then the client secretly writes
  the link $\link[i] := \mathfrak{p}_{\rm real}$. $\link[i]$
  represents secret shares of the next pointers for a real
  element from the first list. The cases for
  $\mathfrak{p}'_{\rm real}, \mathfrak{p}_{\rm dummy}$, and
  $\mathfrak{p}'_{\rm dummy}$ are analogous.

  At the end of this reverse scan, the client remembers
  the position tuple for the first real and first dummy elements of
  both
  linked lists. They are denoted by $\mathfrak{p}^1_{\rm real}$,
  $\mathfrak{p}'^1_{\rm real}$, $\mathfrak{p}^1_{\rm dummy}$, and
  $\mathfrak{p}'^1_{\rm dummy}$.

\item We next call
${\sf Permute}$
to the combined layout $\{\iperm, (\U_b, \link_b)\}_{b \in \zthree}$,
where each server $\S_b$ has input $\pi_{b+1}$,
to produce $\{\pi_b, (\U'_b, \link'_b)\}_{b \in \zthree}$  as
output.

\item The linked lists can now be accessed using the four
  position tuples $\mathfrak{p}^1_{\rm real}$,
  $\mathfrak{p}'^1_{\rm real}$, $\mathfrak{p}^1_{\rm dummy}$, and
  $\mathfrak{p}'^1_{\rm dummy}$. The client first starts accessing real
  elements in the
  two lists using $\mathfrak{p}^1_{\rm real}$ and
  $\mathfrak{p}'^1_{\rm real}$ to merge them. When a real list
  ends, it starts
  accessing the corresponding dummy list.

  More precisely, the client secretly writes the merged result to the abstract
  output array $\U''$.

  Start with $k \gets 0$, $\mathfrak{p}^1 \gets \mathfrak{p}^1_{\rm real}$,
  $\mathfrak{p}^2 \gets \mathfrak{p}^2_{\rm real}$.

  For each $s \in \{1,2\}$,
  the client reconstructs $\mathsf{B}^s := \oplus_b \U'_b[\mathfrak{p}^s_b]$
  and $\mathsf{next}^s := \oplus_b \link'_b[\mathfrak{p}^s_b]$ at most once,
  i.e., if $\mathsf{B}^s$ and $\mathsf{next}^s$ have already been reconstructed once
  with the tuple $(\mathfrak{p}^p_b: b \in \zthree)$,
  then they will not be reconstructed again.

  If $\mathsf{B}^1$ should appear before $\mathsf{B}^2$,
  then the client secretly writes $\U''[k] \gets \mathsf{B}^1$ and updates $k \gets k+1$,
  $\mathfrak{p}^1 \gets \mathsf{next}^1$;
  if the end of the real list is reached, then it updates
  $\mathfrak{p}^1 \gets \mathfrak{p}^1_{\rm dummy}$.
  The case when $\mathsf{B}^2$ should appear before $\mathsf{B}^1$
  is analogous.

The next element is processed until the
client has secretly constructed the whole abstract array $\U''$.

\item The new merged layout $\{\iperm, \U''_b\}_{b \in \zthree}$ is produced.

\end{enumerate}

\begin{theorem}\label{thm:merge}
The ${\sf Merge}$ protocol perfectly securely realizes the ideal
functionality $\mcal{F}_{\rm merge}$ in the presence of a
semi-honest adversary corrupting a single server
with $O(n)$ bandwidth.

\end{theorem}

\ignore{\begin{lemma}
\label{thm:merge}
The above ${\sf Merge}$ protocol
is perfectly secure as per Definition~\ref{secdef}.
\end{lemma}
}

\begin{proof}
We follow the same strategy as in Theorem~\ref{thm:compaction}.
Again, from the construction, the protocol performs merging correctly and re-distributes secretes
using independent randomness.  Hence, the marginal distribution of
the outputs is the same for both the protocol and the ideal functionality.

We fix the inputs of all servers, and some $b \in \zthree$. 
Recall that the goal is to show that (1) the $\view^b$ follows a distribution
that is totally determined by the inputs $\mathbf{I}_b$ and the outputs $\mathbf{O}_b$ of $\S_b$;
(2) conditioning on $\Out_b$,
$\view^b$ is independent of the outputs of the client and all other servers.

The second part is easy, because the inputs are fixed.
Hence, conditioning on $\Out_b$ (which includes $\U''_b$ and $\U''_{b+1}$),  
$\U''_{b+2}$ has no more randomness and totally determines the outputs of other servers.
%
%

To prove the first part, our strategy is to decompose $\view^b$ into a list components,
and show that fixing $\mathbf{I}_b$ and  conditioning on  $\mathbf{O}_b$ and a prefix 
of the components, the distribution of the next component can be determined.
Hence, this also gives the definition of a simulator.

First, observe that in the last step, the client
re-distributes the shares of $\U''$, and gives
output~$\mathbf{O}_b$ (including $\U''_b$ and $\U''_{b+1}$) to $\S_b$;
moreover, the shares of $\U''$ are generated with fresh independent
randomness.  Hence, the distribution of the part of $\view^b$ excluding $\mathbf{O}_b$ is independent of $\mathbf{O}_b$.
We consider the components of $\view^b$ in the following order,
which is also how a simulator generates a view after seeing $\Inp_b$.

\begin{enumerate}

\item \emph{Communication Pattern.}
Observe that from the description of the algorithm,
the communication pattern between the client and the servers
depends only on the length~$n$ of the input arrays.

\item \emph{Random permutation $\pi_{b+1}$.}  This is independently generated
using fresh randomness.

\item \emph{Link $\link$ creation.}  The (abstract) array $\link$
is created by reverse linear scan.  The shares $\link_b$ and $\link_{b+1}$
received by $\S_b$ are generated by fresh independent randomness.

\item \emph{$\mathsf{Permute}$ subroutine.}  By the correctness
of the $\mathsf{Permute}$, the shares of the outputs $(\U', \link')$ 
received by $\S_b$ follow an independent and uniform random distribution.
By the perfect security of 
$\mathsf{Permute}$, the component of $\view^b$ due to
$\mathsf{Permute}$ depends only on the inputs (which include the shares
of $\Arr$ and $\link$) and the outputs of the subroutine.

\item \emph{List traversal.} Since $\S_b$ does not know $\pi_b$ (which is generated
using independent randomness by $\S_{b-1}$),
from $\S_b$'s point of view,
while the elements from the two underlying lists are being merged,
its array $\U'_b$ is traversed in an independent and uniform random order.
\end{enumerate}

Therefore, we have described a simulator procedure that
samples the $\view^b$ step-by-step, given $\mathbf{I}_b$
and $\mathbf{O}_b$.

\paragraph{Efficiency.} The analysis for the ${\sf Merge}$ protocol is similar to that for
  Theorem~\ref{thm:compaction} except that the operations are
  performed on lists of size $2n$ instead of $n$.

\end{proof}

\section{Three-Server One-Time Oblivious Memory}

We construct an abstract datatype  to process
non-recurrent memory lookup 
requests, i.e., 
between rebuilds of the data structure, each distinct address is requested at most once.
Our abstraction is similar to the perfectly secure one-time oblivious 
memory by Chan et
al.~\cite{perfectopram}.
However, while Chan et
al. only consider perfect security with respect to access pattern, 
our three-server one time memory in addition 
information-theoretically encrypts the data itself. 
Thus, in \cite{perfectopram}, since the algorithm does not provide
guarantees for the data itself, it can modify the data structure
while performing operations.
In contrast, our one-time oblivious memory is a read-only data
structure. In this data structure, we
assume every request is tagged with a position label indicating
which memory location to lookup in each of the
servers. In this section, we assume that such a position is
magically available during lookup; but in subsequent sections we
show how this data structure can be maintained and provided
during a lookup.
 
\subsection{Definition: Three-server One-Time Oblivious Memory}
\label{sec:defin-multi-serv}

Our (three-server) one-time oblivious memory supports three
operations: 1) {\sf Build}, 2) {\sf Lookup}, and 3) {\sf
  Getall}. {\sf Build} is called once 
 upfront to create the data structure: it takes in a set of data
 blocks (tagged with its logical address), permutes shares of the
 data blocks at each of the servers to create a data structure
 that facilitates subsequent lookup from the servers. Once the
 data structure is built, {\sf lookup} 
 operations can be performed on it. Each lookup request consists
 of a logical address to lookup and a position label for each of
 the three 
 servers, thereby enabling them to perform the lookup
 operation. The lookup 
 can be performed for a real logical address, in which case the
 logical address and the 
 position labels for each of the three servers are provided; or
 it can be a 
 dummy request, in which case 
 $\bot$ is provided.
Finally, a {\sf Getall } operation is called to obtain a list $\U$ of
all the 
 blocks that were provided during the {\sf Build}
 operation. Later, in our 
 ORAM scheme, the elements in the list $\U$ will be combined with those in other lists
to construct
a potentially larger one-time oblivious memory. 

Our three-server one-time oblivious memory maintains
obliviousness as long as 1) for each real block in the one-time
memory, a lookup is 
performed at most once, 2) at most $n$ total lookups (all of
which could potentially be dummy lookups) are performed, and 3) no two
servers collude with each other to learn the shares of the other
server. 

\subsubsection{Formal Definition.}
\label{sec:formal-definition-motm}
Our three-server one-time oblivious memory scheme $\TOTM[n]$ is
parameterized by $n$, the number of memory lookup requests
supported by the data structure. It is comprised of the following
randomized, stateful algorithms: 

\begin{itemize}[leftmargin=5mm]
\item $\Big(U, \big(\big\{\pi_b, (\widehat{\Arr_b},
  \widehat{\link_b})\big\}_{b \in \zthree}, {\sf dpos}\big)\Big) 
  \leftarrow {\sf Build}(\Arr, \bot)$: 
  \begin{itemize}[leftmargin=3mm]
  \item \emph{Input: } A sorted, unpermuted layout denoted
    $\{\iperm, \Arr_b\}_{b \in \zthree}$ representing an
    abstract sorted list $\Arr$. $\Arr[i]$ represents a key-value pair
    $(\key_i, 
    v_i)$ which
    are either real and contains a real address $\key_i$ and value
    $v_i$, or dummy and contains a $\bot$. The list
    $\Arr$ is sorted by the key $\key_i$. 
    The client's
    input is $\bot$.
  \item \emph{Functionality: } 
    The {\sf Build} algorithm creates a layout
    $\{\pi_b, (\widehat{\Arr_b}, \widehat{\link_b})\}_{b \in
    \zthree}$ of size $2n$ 
    that will facilitate subsequent lookup requests;
		intuitively, $n$ extra dummy elements are added,
		and the $\widehat{\link_b}$'s maintain a singly-linked
		list for these $n$ dummy elements.
		Moreover, the tuple of head positions
		is secret-shared $\oplus_b {\sf dpos}_b$ among the three servers.

    It also outputs a sorted list $U$ of $n$ key-value pairs $(\key,
    \pos)$ 
    sorted by $\key$
    where each $\pos := (\pos_0, \pos_1, \pos_2)$;
		the invariant is that if $\key \neq \bot$,
		then the data for $\key$ is $\oplus_b
                \widehat{\Arr_b}[\pos_b]$. \kartik{I changed the
                  build interface to include dpos too; similarly
                  I changed the lookup interface too!}

    The output list $U$ is stored as a sorted, unpermuted layout
    $\{\iperm, U_b\}_{b \in \zthree}$. 
    Every real key from $\Arr$ appears exactly once in $U$ and
    the remaining 
    entries of $U$ are $\bot$'s.
    The client's output is $\bot$.

    Later in our scheme, $U$ will be propagated back to the 
		corresponding data structure with preceding
		recursion depth during a coordinated rebuild. Hence, $U$ does
    not need to carry the value $v_i$'s. 
\end{itemize}
\item $v \gets {\sf Lookup}\Big(\big(\big\{\pi_b, (\widehat{\Arr_b},
    \widehat{\link_b})\big\}_{b 
    \in \zthree}, \dpos\big),\big(\key,\pos\big)\Big)$: 
  \begin{itemize}[leftmargin=3mm]
  \item \emph{Input: } The client provides a key $\key$ and a position label 
    tuple $\pos := (\pos_0, \pos_1, \pos_2)$. The servers input
    the data structure $\{\pi_b, (\widehat{\Arr_b},
    \widehat{\link_b})\}_{b 
    \in \zthree}$ and ${\sf dpos}$ created during {\sf Build}.

  \item \emph{Functionality: } If $\key \neq \bot$, return
    $\oplus_b \widehat{\Arr_b}[\pos_b]$ 
    else, return
    $\bot$. 
  \end{itemize}

\item $R \gets {\sf Getall}\Big(\big\{\pi_b, (\widehat{\Arr_b},
    \widehat{\link_b})\big\}_{b 
    \in \zthree}, \bot\Big)$: 
  \begin{itemize}[leftmargin=3mm]
  \item \emph{Input: } The servers input
    the  data structure $\{\pi_b, (\widehat{\Arr_b},
    \widehat{\link_b})\}_{b 
    \in \zthree}$ created during {\sf Build}.
  \item \emph{Functionality: }
  the {\sf Getall} algorithm returns a sorted, unpermuted layout
  $\{\iperm, R_b\}_{b \in \zthree}$ 
  of length 
  $n$. This layout represents an  abstract sorted 
  list $R$ of key-value pairs where each entry is
  either real and of the form $(\key, v)$ or dummy and of the form
  $(\bot, \bot)$.
  The list $R$ contains all real elements inserted
  during ${\sf Build}$ including those that have been looked up,
  padded with $(\bot, \bot)$ to a length of $n$.\footnote{The
    {\sf Getall} function returns as output the unpermuted layout
    that was input to {\sf Build}. It primarily exists for ease
    of exposition.}
\end{itemize}
\end{itemize}
\paragraph{Valid request sequence.}
Our three-server one-time oblivious memory ensures obliviousness
only if lookups are non-recurrent (i.e., the same real key is never looked up more than once); and the number of lookups is upper bounded by $n$, the size
of the input list provided to ${\sf Build}$. More formally,
a sequence of operations is valid, iff the following holds:

\begin{itemize}
\item The sequence begins with a single call to {\sf Build}, followed
by a sequence of at most $n$ {\sf Lookup} calls, and finally the sequence
ends with a call to {\sf Getall}.

\item All real keys in the input provided to {\sf Build} have distinct
keys.

\item For every ${\sf Lookup}$ concerning a real element with client's input $(\key, \pos := (\pos_0, \pos_1, \pos_2))$, the $\key$ should have 
existed in the input to {\sf Build}. Moreover, the position label
tuple $(\pos_0, \pos_1, \pos_2)$  must be the correct position labels
for each of the three servers.

\item No two {\sf Lookup} requests should request the same real key.
\end{itemize}

\paragraph{Correctness.} 
Correctness requires that:
\begin{enumerate}
\item For any valid request sequence, with probability 1, every 
{\sf Lookup} request must return the correct value $v$ associated
with key $\key$ that was supplied in the {\sf Build} operation.

\item For any valid request sequence, with probability 1, {\sf Getall}
must return an array $R$ containing every $(\key, v)$ pair
that was supplied to {\sf Build}, padded with dummies to have 
$n$ entries.\end{enumerate}

\paragraph{Perfect obliviousness.}  Suppose the following
sequence of operations are executed: the initial ${\sf Build}$,
followed by a valid request sequence of $\ell$ ${\sf Lookup}$'s,
and the final ${\sf Getall}$.  Perfect obliviousness
requires that for each $b \in \zthree$,
the joint distribution of the communication pattern
(between the client and the servers) and the $\view^b$ of $\S_b$
is fully determined by the parameters~$n$ and~$\ell$.


\subsection{Construction}
\label{sec:construction}

\subsubsection{Intuition.}
\label{sec:intuition}

The intuition is to store shares of the input list on storage
servers such that each share is independently permuted 
and each server storing a share
does not know its
permutation (but some other server does).
In order to lookup a real element, if a  position label
for all three shares are provided, then the client can directly
access the shares. 
Since the shares are permuted and the server storing a share does
not know the permutation, each lookup
corresponds to accessing a completely random location and is thus
perfectly oblivious.
This is true so far as each element is accessed exactly once and
the position label provided is correct; both of these constraints
are satisfied by a valid request sequence.
However, in an actual request sequence, some of the requests may
be dummy and these requests do not carry a position label with
them.
To accommodate dummy requests, before permuting the shares, we
first append 
shares of dummy elements to shares of the unpermuted input list.
We add dummy elements enough to support all lookup requests
before the one time memory is destroyed.
Then we create a linked list of dummy elements so that a  dummy
element stores the 
position label of the location where the next dummy element is
destined to be after permutation. The client maintains
the head of this linked list, updating it every time a dummy
request is made.
To ensure obliviousness, the links (position
labels) in the dummy linked list are also stored
secret-shared and permuted along with the input list. \anti{refine this}
\ignore{
For dummy elements, we append dummy elements enough to
support potentially all dummy lookup requests after the one time
memory is destroyed.
In order to lookup a real
element, we assume that a correct position label for all three
shares are 
provided; this is ensured by the constraints on a valid request
sequence for the one time memory.
Since the shares are permuted and the server storing a share does
not know the permutation, for each of the servers, each lookup
corresponds to accessing a completely random location and is thus
perfectly oblivious.
However, lookups for dummy requests is slightly tricky as they do
not have a position label associated 
with them. Thus, we maintain an (abstract) oblivious linked list
of  dummies so that each dummy can point to the next dummy using
position labels. Thus the client can access a dummy element by
just accessing the head of the list and then updating the head to
the next element. To ensure obliviousness, the links (position
labels) in the dummy linked list are 
secret-shared and permuted along with the input list.
}

\subsubsection{Protocol}
\label{sec:deta-constr}

\underline{\textbf{Build.}}
Our oblivious {\sf Build} algorithm proceeds as follows. Note
that the input list $\Arr$ is stored as an unpermuted layout
$\{\iperm, \Arr_b\}_{b \in \zthree}$ on the three servers. 

\begin{enumerate}
\item \emph{Initialize to add dummies.} 
Construct an extended abstract  $\Arr'[0..2n-1]$ of length $2n$ such that
the first $n$ entries are key-value pairs copied from the input~$\Arr$ (some of which 
may be dummies). 

The last $n$ entries of $\Arr'$ contain \emph{special}
dummy keys. For each $i \in [1..n]$, the special dummy key~$i$ is stored in $\Arr'[n-1+i]$,
and the entry has a key-value pair denoted by $\bot_i$.
For each~$i \in [1..n]$, the
client secretly writes $\bot_i$ to $\Arr'[n-1+i]$.


\item \emph{Generate permutations for $\OTM$.} Each server $\S_{b}$ acts
  as the permutation server for $\S_{b+1}$. For each 
  $b \in \zthree$, the client generates a random permutation
  $\pi_{b+1}: [2n] \rightarrow [2n]$ on permutation server 
  $\S_{b}$. 

\item \emph{Construct a dummy linked list.} Using the newly generated
permutation $\pi_{b+1}$ on server $\S_b$, the client constructs
a linked list of dummy blocks. This is to enable accessing the
dummy blocks linearly, i.e.,
for each $i \in [1..n-1]$,
after accessing dummy block $\bot_i$,
the client should be able to access $\bot_{i+1}$.

 The client simply leverages $\pi_{b+1}(n .. 2n-1)$ stored on 
server $\S_b$ to achieve this.
Specifically, 
for $i$ from $n-1$ down to 1,
to create a link between
$i$-th  and $(i+1)$-st dummy,  the
client reads $\pi_{b+1}(n + i)$ from server $\S_b$ and secretly
writes the tuple $(\pi_{b+1}(n+i) : b \in \zthree)$
to the abstract link $\link[n+i-1]$.

There are no links between real elements,
 i.e., for $j \in [0..n-1]$,
the client secretly writes $(\bot,
\bot, \bot)$ to (abstract) $\link[j]$.

Observe that these
links are secret-shared and stored as an unpermuted layout
$\{\iperm, \link_b\}_{b \in \S_b}$.

Finally, the client records the positions
of the head of the lists
and secretly writes the tuple across the three servers,
i.e.,
$\oplus_b {\sf dpos}_b := (\pi_b(n): b \in \zthree)$,
where ${\sf dpos}_{b}$ is stored on server $\S_b$.


\item \emph{Construct the key-position map $U$}. The client can
construct the (abstract) key-position map $U[0..n-1]$ sorted by
the key from the 
first $n$ entries of $\Arr'$ 
and the $\pi_{b}$'s.
Specifically, for each $i \in [0..n-1]$,
the client secretly writes 
$(\key_i, (\pi_b(i): b \in \zthree))$
to $U[i]$.

Recall that
$U$ is stored as a sorted, unpermuted layout $\{\iperm, U_b\}_{b \in \zthree}$. 

\item \emph{Permute the lists along with the links.} Invoke
  {\sf Permute} with input $\{\iperm,
  (\Arr'_b, \link_b)\}_{b \in \zthree}$, and permutation $\pi_{b+1}$ as the
  input for $\S_b$.  
  The ${\sf Permute}$ protocol returns a permuted output layout
  $\{\pi_b, 
  (\widehat{\Arr_b}, \widehat{\link_b})\}_{b \in \zthree}$. 

\item As the  data structure, each server $\S_b$
  stores $(\widehat{\Arr_b}, \widehat{\link_b})$, $(\pi_{b+1},
  (\widehat{\Arr_{b+1}}, \widehat{\link_{b+1}}))$, and ${\sf
    dpos_{b+1}}$. 
  The algorithm returns key-position map list $U$ as output,
  which is stored as 
 an unpermuted layout $\{\iperm, U_b\}_{b \in \zthree}$.
This list will later be passed to the preceding recursion depth 
in the ORAM scheme during a coordinated rebuild operation.
\end{enumerate}

\begin{fact}
  \label{fact:otm-build}
  The {\sf Build} algorithm for building an $\TOTM$ supporting
  $n$ lookups requires an $O(n)$ bandwidth. 
\end{fact}

\begin{proof}
  Each of the steps in the protocol either generate a random permutation
  using Fisher-Yates, invoke {\sf Permute}, or linearly scan
  lists of size $O(n)$ blocks. Each of these steps can be
  performed with an $O(n)$ bandwidth.
\end{proof}
\noindent \textbf{Protocol} \underline{\textbf{Lookup.}} Our oblivious ${\sf
  Lookup}\Big(\big(\big\{\pi_b, (\widehat{\Arr_b},
    \widehat{\link_b})\big\}_{b 
    \in \zthree}, \dpos\big),\big(\key, (\pos_0, \pos_1, \pos_2)\big)\Big)$ algorithm proceeds as follows: 

\begin{enumerate}

\item The client reconstructs $(\pos'_0, \pos'_1, \pos'_2)
\gets \oplus_b {\sf dpos}_b$.

\item \emph{Decide position to fetch from.} If $\key \neq \bot$, set
$\pos \gets (\pos_0, \pos_1, \pos_2)$,
i.e., we want to use the position map supplied from the input;
if $\key = \bot$,
set $\pos \gets (\pos'_0, \pos'_1, \pos'_2)$,
i.e., the dummy list will be used.

\item \emph{Reconstruct data block.}
Reconstruct $v
  \gets \oplus \widehat{\Arr_b}[\pos_b]$
	and $(\widehat{\pos}_0, \widehat{\pos}_1, \widehat{\pos}_2)
	\gets \oplus \widehat{\link_b}[\pos_b]$.

\item \emph{Update head of the dummy linked list.}
If $\key \neq \bot$,
the client  re-shares the secrets
 $\oplus_b {\sf dpos}_b \gets (\pos'_0, \pos'_1, \pos'_2)$ with the same head;
if $\key = \bot$,
the client secretly shares the updated head
$\oplus_b {\sf dpos}_b \gets (\widehat{\pos}_0, \widehat{\pos}_1,
\widehat{\pos}_2)$. 

  \ignore{
  If this was the $i$-th
  lookup request after rebuild, if $i < n$, for each server
  $\S_b$ set ${\sf dpos}_{b+1} = \pi_{b+1}[n - 1 +
  i+1]$. \kartik{maintain i? Make lookup stateful?}}

\item \label{step2} \emph{Read value and return.} Return
  $v$.
	
\end{enumerate}

\begin{fact}
  \label{fact:otm-lookup}
  The {\sf Lookup} algorithm for looking up a block in $\TOTM$
  requires $O(1)$ bandwidth.
\end{fact}

\noindent {\textbf{Protocol} \underline{\textbf{Getall.}} For ${\sf Getall}$, the
client simply 
invokes the {\sf Unpermute} protocol 
on input layout $\{\pi_b, 
(\widehat{\Arr_b}, \widehat{\link_b})\}_{b \in \zthree}$ and
returns the first $n$ entries of the 
sorted, unpermuted layout (and ignores the links created). 
This output is also stored as a sorted, unpermuted layout
$\{\iperm, \Arr_b\}_{b \in \zthree}$.
The data structure created  
on the servers during {\sf Build} can now be destroyed.
\begin{fact}
  \label{fact:otm-lookup}
  The $\TOTM$ {\sf Getall} algorithm 
  requires an $O(n)$ bandwidth.
\end{fact}

\begin{lemma}
\label{lemma:44}
The subroutines ${\sf Build}$, 
${\sf Lookup}$ and ${\sf Getall}$ are correct and perfectly oblivious
in the presence of a
semi-honest adversary corrupting a single server.
\end{lemma}

\begin{proof}
We go through each subroutine and explain its purpose,
from which correctness follows;
moreover, we also argue why it is perfectly
oblivious. 
Fix some $b \in \zthree$, and we consider the $\view^b$ of corrupt $\S_b$.

\noindent \textbf{Build subroutine.}  We explain each step, and why
the corresponding $\view^b$ satisfies perfect obliviousness.

\begin{enumerate}
\item \emph{Initialize to add dummies.} 
This step appends $n$ dummies at the end of the abstract array.
The data access pattern depends only $n$.  The contents of the extra dummy entries
are secretly shared among the servers using independent randomness.

\item \emph{Generate permutations for $\OTM$.} 
The client generates an independent random permutation for each server.
In particular, the permutation $\pi_{b+1}$ received by $\S_b$ is independent and uniformly random.

\item \emph{Construct a dummy linked list.} 
In this step, a linked list is created for the dummy entries.
Using $\pi_{b+1}$, server~$\S_b$ observes a linear scan on its array
and creates a linked list that is meant for $\S_{b+1}$.
Similarly, the linked list for $\S_b$ is created by $\S_{b-1}$.

The links themselves are secretly shared and stored in the abstract $\link$,
and so are the head positions~$\oplus_b {\sf dpos}_b$.
Hence, the data seen by $\S_b$ appears like independent randomness.

\item \emph{Construct the key-position map $U$}. 
This step is just a linear scan on the abstract $\Arr'$ and
each permutation array $\pi_b$.
The resulting abstract array $U$ is also secretly shared.

\item \emph{Permute the lists along with the links.}
This steps uses the building block ${\sf Permute}$,
which is proved to be perfectly secure in Lemma~\ref{thm:permute}.

\item Finally, as guaranteed by ${\sf Permute}$,
the resulting arrays $\widehat{\Arr}$ and $\widehat{\link}$ are secretly shared 
using independent randomness.  Moreover, the abstract list $U$
is returned as required.  Hence, the ${\sf Build}$ subroutine 
is correct and perfectly oblivious, as required.

\end{enumerate}

\noindent \textbf{Lookup Subroutine.} By construction,
each call to ${\sf Lookup}$ is supplied with the correct position labels,
and hence, correctness is achieved.  We next argue
why the distribution of the $\view^b$ of corrupt $\S_b$ depends only on
the number $\ell$ of lookups.

Observe that the entries accessed in $\S_b$ are permuted randomly by $\pi_b$, which is unknown to
$\S_b$.  Since the request sequence is non-recurrent, each real key is requested at most once.
On the other hand, the dummy entries are singly linked, and each dummy entry is accessed also at most once.
Hence, the $\ell$ lookups correspond to $\ell$ distinct uniformly random accesses in $\S_b$'s corresponding arrays.

\noindent \textbf{Getall Subroutine.}  This subroutine calls the building block ${\sf Unpermute}$,
which is proved to be perfectly secure in Lemma~\ref{thm:unpermute}.

As mentioned before, all data stored on $\S_b$ is secretly shared using independent randomness.
Hence, the distribution of the overall $\view^b$ depends only on the number $\ell$ of lookups.

\end{proof}

\ignore{
The proof can be found in the supplemental
material, Section~\ref{proof:OTM}.
}


\section{3-Server ORAM with $O(\log^2 N)$ Simulation Overhead}
\label{sec:posoram}


\subsection{Position-Based ORAM}
\label{sec:position-based-oram}
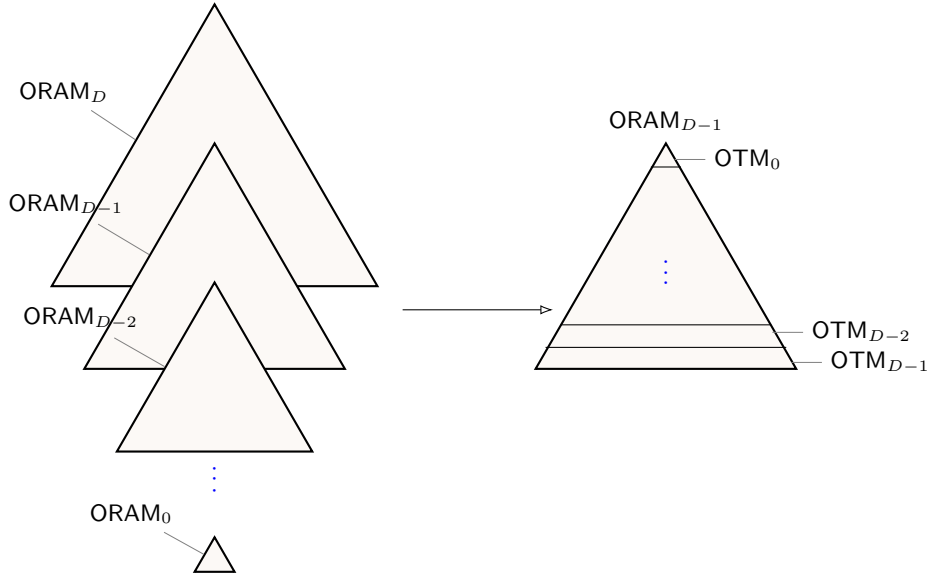
\begin{figure}[!tbp]

  \tikzset{
  triangle/.style = {draw=black, fill=brown!5, thick},
  squiggle/.style = {decoration={snake, segment length=5mm}, decorate}
}

\begin{tikzpicture}

   \centering
  \node[shape=regular polygon,regular polygon sides=3,draw=black, fill=brown!5, thick,
minimum size=5cm,anchor=south,pin=140:$\ORAM_D$] (A1) at (0,4.3) {};
\node[shape=regular polygon,regular polygon sides=3,draw=black, fill=brown!5, thick,
minimum size=4cm,anchor=south,pin=140:$\ORAM_{D-1}$] (A2) at (0,3.2) [above]{};
\node[shape=regular polygon,regular polygon sides=3,draw=black, fill=brown!5, thick,
minimum size=3cm,anchor=south,pin=140:$\ORAM_{D-2}$] (A3) at (0,2.1) [above]{};
\node[shape=regular polygon,regular polygon sides=3,draw=black, fill=brown!5, thick,
minimum size=0.5cm,anchor=south,pin=140:$\ORAM_0$] (A4) at (0,0.5) [above]{};
  

 \node[align=center] (n0)  at (0,2.2)  {$$};
 \node[align=center] (n1)  at (0,1.3)  {$$};
 \path (n0) -- (n1) node [blue, midway, sloped] {$\ldots$};

\draw[-{Latex[open]}] (2.5,4) -- ++(2,0);

  \node[shape=regular polygon,regular polygon sides=3,draw=black, fill=brown!5, thick,
minimum size=4cm,anchor=south,label=above:$\ORAM_{D-1}$] (B2) at (6,3.2) [above]{};
  
\draw (4.4,3.5) -- (7.6,3.5);
 \node[pin=0:$\OTM_{D-1}$] (a1)  at (7.55,3.3)  {$$};

\draw (4.6,3.8) -- (7.4,3.8);
 \node[pin=0:$\OTM_{D-2}$] (a2)  at (7.3,3.7)  {$$};

\draw (5.8,5.9) -- (6.2,5.9);
 \node[pin=0:$\OTM_0$] (a3)  at (6,6)  {$$};

 \node[align=center] (n2)  at (6,3.5)  {$$};
 \node[align=center] (n3)  at (6,5.5)  {$$};
 \path (n2) -- (n3) node [blue, midway, sloped] {$\ldots$};

\end{tikzpicture}

 \caption{ORAM Construction from Position-Based ORAM
 }\label{fig}
\end{figure}

Our ORAM scheme will consist of logarithmically many
position-based ORAMs of geometrically increasing sizes,
henceforth denoted $\ORAM_0$, $\ORAM_1$, $\ldots$, $\ORAM_D$ where
$D := \log_2 N$. See Figure~\ref{fig} for a high level sketch of our construction. 

 Specifically, $\ORAM_d$ stores
$\Theta(2^d)$ blocks where $d \in \{0, 1, \ldots, D\}$.
The actual data blocks are stored in $\ORAM_D$ whereas all
other $\ORAM_d, d < D$ recursively stores position labels
for the next depth $d+1$.

In this subsection, we focus on describing $\ORAM_d$ assuming the
position labels are magically available. In the next 
subsection, we will describe how position labels are maintained
across different depths.

\subsubsection{Data Structure.}
\label{sec:data-structure}

For $0 \leq d \leq D$
each $\ORAM_d$ consists of $d+1$ levels
of three-server one-time oblivious memory 
that are geometrically increasing in size. We denote these
one-time oblivious memories as $(\TOTM_j : j = 0, \ldots, d)$
where $\TOTM_j := \TOTM[2^j]$ stores at most $2^j$ real blocks.

Every level $j$ is marked as either \emph{empty} (when the
corresponding $\TOTM_j$ has not been built) or \emph{full}
(when $\TOTM_j$ is ready and in operation). Initially, all
levels are empty.

\paragraph{Position label.} 
To access a block stored in $\ORAM_d$,
its position label specifies
1) the level $l \in [0..d]$ such that the block resides in
$\TOTM_\ell$; and
2) the tuple $\pos := (\pos_0, \pos_1, \pos_2)$ to 
reconstruct the block from $\TOTM_\ell$.

\subsubsection{Operations}
\label{sec:operations}

Each position-based ORAM supports two operations, {\sf Lookup}
and {\sf Shuffle}.

\noindent\textbf{Protocol }\underline{\textbf{Lookup:}}

\begin{itemize}
\item \emph{Input:} The client provides $\big(\key, \pos := (l, (\pos_0, \pos_1, \pos_2))\big)$ as input,
where $\key$ is the logical address for the lookup request, $l$ represents the level
such that the block is stored in $\TOTM_l$, and $(\pos_0, \pos_1, \pos_2)$ is 
used as an argument for $\TOTM_l.{\sf Lookup}$.

The servers store $\TOTM_j$ for $0 \leq j \leq d$ where
$\TOTM$ stores layout $\big\{\pi_b, (\widehat{\Arr_b},
    \widehat{\link_b})\big\}_{b 
    \in \zthree}$ and $\dpos$ for the level. Moreover, some of the OTMs may be empty.

\item \emph{Algorithm:} The lookup operation proceeds as follows:

\begin{enumerate}
\item\noindent For each non-empty level $j = 0, \ldots, d$, perform the following:
  \begin{itemize}
  \item The position label specifies that the block is stored at level $\TOTM_l$.
    For level $j = l$, set $\key' := \key$ and $\pos' := (\pos_0, \pos_1, \pos_2)$. For all other levels, set $\key' := \bot$, $\pos' := \bot$.
    \item $v_j \gets \TOTM_j.{\sf Lookup}\Big(\big(\big\{\pi_b, (\widehat{\Arr_b},
    \widehat{\link_b})\big\}_{b 
    \in \zthree}, \dpos \big), (\key', \pos')\Big)$.
    
  \end{itemize}
  \item Return $v_l$.
\end{enumerate}
\end{itemize}

\begin{fact}
\label{fact:lookup_pos}
  For $\ORAM_d$, ${\sf Lookup}$ requires an $O(d)$ bandwidth.
\end{fact}

\noindent\textbf{Protocol} \underline{\textbf{Shuffle.}} 
The shuffle operation is used in hierarchical ORAMs to shuffle
data blocks in consecutive smaller levels and place them in the
first empty level (or the largest level). 
Our shuffle operation, in addition, accepts another input $U$ that
is used to update the contents of data blocks stored in the position
based ORAM. In the final ORAM scheme, the list $U$ passed as an input to
$\ORAM_d$ will contain the (new) position labels of blocks in $\ORAM_{d+1}$.
Similarly, the shuffle operation returns an output $U'$ that will be
passed as input to $\ORAM_{d-1}$. \anti{we should have defined shuffle formally as we did with the rest, if there is time we should do it}
More formally, our shuffle operation can be specified as follows:



$(U',\widehat{\Arr} ) \leftarrow {\sf
  Shuffle}_d\big((\OTM_0,\ldots, \OTM_l, U), l \big)$:

\kartik{I modified the input, $\widehat{\Arr}$ as input did not
  make sense. The servers input OTMs!}
\begin{itemize}[leftmargin=5mm]
\item \emph{Input:} 
The shuffle operation for $\ORAM_d$ accepts as input from the
client a level $l$
in order to build $\OTM_l$ 
from data blocks currently in levels $0, \ldots, l$. In addition,
$\ORAM_d$ consists of an extra $\OTM$, denoted by $\OTM'_{0}$,
containing only a single element. Jumping ahead, this single
element represents a freshly fetched block.   

The inputs of the servers consist of $\OTM$s for levels up to
level $l$, each of which is
stored as a permuted layout $\{\pi_b, (\widehat{\Arr}_b, \widehat{\link_b})\}_{b \in
  \zthree}$ 
and an array of key-value pairs $U$, stored as a sorted,
unpermuted layout $\{\iperm, U_b\}_{b \in \zthree}$. The array
$U$ is used to update the blocks during the shuffle operation.  

Throughout the shuffle operation we maintain the following invariant:
\begin{itemize}[leftmargin=5mm]
\item For every $\ORAM_d$, $l \leq d$. Moreover,
either level $l$ is the smallest empty level of $\ORAM_d$ or $l$
is the largest level, i.e., $l = d$. 

\item Each logical address appears at most once in $U$.

\item The input $U$ contains a subset of logical addresses that appear in levels 
$0,\ldots,l$ of the $\ORAM_d$ (or $\OTM'_{0}$). 

Specifically, given a key-value pair $(\key,v)$, the corresponding block $(\key,v')$ \anti{in $U'$: does $v'$ is the value that will be in $U'$, it is confusing with the matrix A} should
already appear in some level in $[0..l]$ or $\OTM'_0$.  An ${\sf update}$ rule
will determine how $v$ and $v'$ are combined to produce a new
value $\widehat{v}$ for~$\key$.

\end{itemize}

\item The Shuffle algorithm proceeds as follows:

\begin{enumerate}[leftmargin=3mm]
\item \textbf{Retrieve key-value pairs from $(\TOTM_0,\ldots,\TOTM_l )$.} 
The client first retrieves the key-value pairs of real blocks from $(\TOTM_0,\ldots,\TOTM_l )$ and restore each array to its unpermuted form.
More specifically, the client constructs the unpermuted sorted
$\Arr^j \gets \TOTM_j.{\sf Getall}(\{\pi_b, (\widehat{\Arr}_b,
\widehat{\link_b})\}_{b \in 
  \zthree},\bot)$, for $0 \leq j \leq l$,\anti{it should be $<$
  instead of $\leq$ but then aren't we missing $\TOTM_0$}
and ${\Arr}^0 \gets \TOTM'_0.{\sf Getall}(\{\pi_b,
(\widehat{\Arr}_b, \widehat{\link_b})\}_{b \in
  \zthree},\bot)$\footnote{The layout inputs to the $\sf Getall$
  operation are restricted to the ones stored in $\TOTM_j$ for $0
  \leq j \leq l$, respectively.}\anti{what do you think about
  this footnote?}
Now, the 
old $\TOTM_0, \ldots, \TOTM_l$ instances can be destroyed.


\item \textbf{Create a list for level $l$.} The client then
  creates a level $l$ list 
of keys from $(\TOTM_0,\ldots,\TOTM_l )$.

\begin{itemize}[leftmargin=2mm]
\item
\emph{Merge lists from consecutive levels to form level $l$ list.} The merge procedure proceeds as follows: 
\begin{itemize}[leftmargin=3mm]

\item[] For $j = 0, \ldots, l-1$ do:

\qquad$\widehat{\Arr}^{j+1} \gets {\sf
  Merge}((\widehat{\Arr}^j,\Arr^j), \bot)$
where $\Arr^j$ and $\widehat{\Arr}^j$ are of size $2^j$

Moreover, the lists are individually sorted
but may contain blocks that have already been accessed.
In the ${\sf Merge}$ protocol, for two elements with the same key 
and belonging to different $\TOTM$ levels, we prefer the one at
the smaller level first.
For the case where $l = d$, perform another merge $\widehat{\Arr}^d \gets {\sf
Merge}((\widehat{\Arr}^d, {\Arr}^d), \bot)$
to produce an array of size $2^{d+1}$; Jumping ahead, the size will be reduced back to $2^d$ in subsequent steps.

\end{itemize}
At the end of this step, we obtain a merged sorted list $\widehat{\Arr}^l$, stored as $\widehat{\Arr}^l := \{\iperm, \widehat{\Arr}^l_b\}_{b \in \zthree}$, containing
duplicate keys that are stored multiple times (with potentially different values).

\item \emph{Mark duplicate keys as dummy.} 
From the stored duplicate keys, we only need the value of the one that 
corresponds to the latest access. All other duplicate entries can
be marked as dummies.
At a high level, this can be performed in a single pass by the client 
by scanning consecutive elements of the unpermuted sorted layout $\widehat{\Arr}^l$.
The client keeps the most recent version, i.e.,
the version that appears first (and has come from the smallest $\TOTM$),
and marks other versions as dummies.
To maintain obliviousness, the secret shares need to be re-distributed for each scanned
entry.

More specifically, suppose that there are $\lambda$ duplicate keys. Then, the client scans through the unpermuted layout $\widehat{\Arr}^l := \{\iperm, \widehat{\Arr}^l_b\}_{b \in \zthree}$.
For consecutive $\lambda$ elements, $j, \ldots, j+\lambda-1$ with the same
key, 
the client
re-distributes the secret for $\widehat{\Arr}^l[j]$ for position $j$,
and secretly writes $\bot$ for positions $j+1, \ldots, j+\lambda-1$.

After this step, the resulting (abstract) $\widehat{\Arr}^l$ is semi-sorted.

\item \emph{Compaction to remove dummies.}
The client invokes the {\sf StableCompact} protocol with input
$\widehat{\Arr}^l := \{\iperm, \widehat{\Arr}^l_b\}_{b \in
\zthree}$, i.e., $\widehat{\Arr}^l \leftarrow{\sf
StableCompact}(\widehat{\Arr}^l, \bot)$ to obtain a sorted, 
unpermuted layout (where the dummies are at the end). We keep the
first $2^l$ entries.
\end{itemize}

\item \textbf{Update $\widehat{\Arr}^l$ with values from $U$.}
  The client 
  updates $\widehat{\Arr}^l$ so that it contains updated position
  values 
from $U$.  Looking ahead, in our final scheme, $U$ will contain
the new position 
labels from an ORAM at a larger depth. 
Given that $\ORAM_D$ is the largest depth and does not store position
values, this step is skipped for $\ORAM_D$.

We do this as follows:
\begin{itemize}[leftmargin=3mm]
\item \emph{Merge $\widehat{\Arr}^l$ with $U$.} The client
  performs $A \gets {\sf Merge}((\widehat{\Arr}^l, U), \bot)$ to obtain a
  sorted, unpermuted layout. Ties on the same key break by
  choosing the blocks in $\widehat{\Arr}^l$.

\item \emph{Scan and Update $A$.} 
In a single pass through the
sorted, unpermuted layout $A$, it can operate on every adjacent
pair of entries. If they share the same key, 
the following {\sf update} rule is used to update both the
values (the precise {\sf update} rule is provided in the ${\sf Convert}$ subroutine in Section \ref{sec:oram-scheme}). In particular, in the final ORAM scheme, the keys in $A$ 
correspond to logical 
addresses. 
Each address in a position-based ORAM
at depth-$d$
stores 
position labels for two children addresses at depth-$(d+1)$.
The entries in $A$ that come from $\widehat{\Arr}^l$ contain the
old position labels for both children. For the entries 
from $U$, if children position labels exist, they correspond to
the new labels.
For each of the child addresses, if $U$ contains
a new position label, the {\sf update} function chooses the new
one; otherwise, it chooses 
the old label from $\widehat{\Arr}^l$.
\ignore{
Scan consecutive elements in the sorted, unpermuted layout $A$. 

If they share the same logical address, 
the {\sf update} rule is used to generate a new value.
The first entry is preserved
and updated to the new value, and the second entry is set to dummy.\anti{Kartik will update these two sentences}
}

\item \emph{Compaction to remove dummies.}
The client invokes the {\sf StableCompact} protocol $A
\leftarrow{\sf StableCompact}(A, \bot)$ to obtain
an updated sorted, unpermuted layout $A$. We keep the first $2^l$ entries. \end{itemize}

\item \textbf{Build $\TOTM_l$.} The client invokes $U' \leftarrow
  {\sf Build}(A, \bot)$ to generate a
 data structure $\TOTM_l$ and $U'$. Mark $\TOTM_l$
as \emph{full} and $\TOTM_i$, for $i < l$, as \emph{empty}.

\end{enumerate}
\end{itemize}

We now prove that the above position-based ORAM is correct and
satisfies perfect obliviousness 
in the presence of a
semi-honest adversary corrupting a single server. 
We also prove that an
  $\ORAM_d$, and $l \leq d$, ${\sf Shuffle}_d$ for level $l$
  requires a bandwidth of
  $O(2^l)$. We use the same
proof framework as \cite{perfectopram}.
However, instead of having only one server, we 
argue perfect obliviousness from an adversary that 
has access to only one of the three servers.

\begin{fact}[Correctness]
Our Position-based ORAM maintains correctness.
More specifically, at every recursion depth $d$, the correct position
labels will be input to the ${\sf Lookup}$ operations
of $\ORAM_d$; and every request
will return the correct answers.
\end{fact}
\begin{proof}
Straightforward by construction.
\end{proof}

For every $\ORAM_d$ at recursion depth $d$, 
the following invariant is respected by construction
as stated in the following facts. 

\begin{fact}[Lifetime of $\OTM$]
\label{fct:otmlife}
For every $\ORAM_d$, every $\OTM_i$ instance at level $i \leq d$
that is created needs to answer 
at most $2^i$ requests 
before $\OTM_i$ instance is destroyed. 
\end{fact}
\begin{proof}
For every $\ORAM_d$, the following is true:
imagine that there is a $(d+1)$-bit binary counter initialized to $0$ that increments
whenever a request comes in.
Now, for $0 \leq \ell < d$, 
whenever the $\ell$-th bit flips from $1$ to $0$, the $\ell$-th level 
of $\OTM_\ell$ is destroyed;
whenever the $\ell$-th bit flips from $0$ to $1$, the $\ell$-th level
of $\OTM_\ell$ is reconstructed. 
For the largest level $d$ of $\ORAM_d$, 
whenever the $d$-th (most significant)
bit of this binary counter flips from $0$ to $1$ or from $1$ to $0$, 
the $(d+1)$-st level is destroyed and reconstructed.
The fact follows in a straightforward manner by observing 
this binary-counter argument.
\end{proof}

\begin{fact}\label{fct:otmnonrecur}
For every $\ORAM_d$ and every $\OTM_\ell$ instance at level $\ell \leq d$, 
during the lifetime of the $\OTM_\ell$ instance:
(a) no two real requests 
will ask for the same key;
and (b) for every request that asks for a real $\key$,
a block with $\key$ must exist in $\OTM_i$.
\end{fact}
\begin{proof}
We first prove claim (a).
Observe that for any $\ORAM_d$, if the block with some 
$\key$
is fetched from 
some level $\ell \leq d$, 
at this moment,
that block
will either enter a smaller level $\ell' < \ell$;
or some level $\ell'' \geq \ell$ will be rebuilt and
the block with $\key$
will go into level $\ell''$ --- in the latter case,
level $\ell$ will be destroyed prior to the rebuilding of level $\ell''$. 
In either of the above cases, 
due to correctness of the construction, if the  
$\key$
is needed again from $\ORAM_d$, 
a correct position label will be 
provided for 
$\key$
such that the request 
will not go to level $\ell$ (until the level is reconstructed).
Finally, claim (b) follows from correctness of the position labels.
\end{proof}

Given the above facts, our construction maintains perfect obliviousness
as in~\cite{perfectopram}.  We emphasize that while the proof structure
is the same, our notion of perfect obliviousness is defined in terms
of the view of one server plus the communication pattern between the client and all the servers.

\begin{lemma}[Obliviousness]
\label{lemma:posoram_obliv}
The above position-based ORAM construction satisfies perfect obliviousness
in the presence of a
semi-honest adversary corrupting a single server.
\end{lemma}

\begin{proof}
We would like to point out the major difference from~\cite{perfectopram}.
Since we do not rely on cryptographic assumptions,
the data content is protected by secret-sharing among the three servers.

For every $\OTM$ instance
constructed during the lifetime of the ORAM,  
Facts~\ref{fct:otmlife}
and \ref{fct:otmnonrecur} are satisfied, and thus every 
one-time memory instance
receives a valid request sequence.
The lemma then follows 
in a straightforward fashion by the perfect obliviousness
of the $\OTM$ scheme (and the perfect security of the underlying building blocks such as ${\sf Merge}$), 
and by observing
that all other access patterns of the ORAM construction are
deterministic
and independent of the input requests.
\end{proof}

%
%
%

\begin{fact}
\label{fact:shuffle_pos}
  Suppose that the {\sf update} function can be evaluated in
  $O(1)$ time.
	Then, for an
  $\ORAM_d$, and $l \leq d$, ${\sf Shuffle}_d$ for level $l$
  requires a bandwidth of
  $O(2^l)$.
\end{fact}
\begin{proof}
  Observe that to obtain real key-value pairs from $\TOTM$'s and
  to construct a level $l$ list, the client invokes {\sf Getall}
  and {\sf Merge} protocols on (abstract) lists of sizes $2^0,
  \ldots, 2^l$, and {\sf StableCompact} on size $O(2^l)$. Since
  each of these steps require linear number of block operations,
	the total number of block operations for these
  steps is $\sum_{i = 1}^{l} (2^i) + O(2^l) = O(2^l)$.
  In order to update the list to contain values from $U$, we
  again invoke {\sf Merge}, {\sf StableCompact}, and an {\sf
    update}  on a list of
  size $O(2^l)$. This requires $O(2^l)$ block operations
  assuming the {\sf update} rule itself requires $O(1)$ operations
  per block.
  Finally, the call to ${\sf Build}$ $\OTM_l$ requires $O(2^l)$
  block operations.
\end{proof}


\subsection{ORAM Construction from Position-Based ORAM}
\label{sec:oram-scheme}

Our ORAM scheme consists of $D+1$ position-based ORAMs denoted as
$\ORAM_0$, $\ldots$, $\ORAM_D$ where
$D = \log_2 N$.
$\ORAM_D$ stores data blocks whereas $\ORAM_d$ for $d < D$ stores
a position map for $\ORAM_{d+1}$.
The previous section specified the construction of a position-based
ORAM. However, it assumed that position labels are magically available
at some $\ORAM_d$. In this section, we show a full ORAM scheme 
and specify 1)~how these position labels for $\ORAM_d$
are obtained from $\ORAM_{d-1}$, and 2)~after a level of $\ORAM_d$ is built, 
how the position labels of blocks from the new level
are updated at $\ORAM_{d-1}$.

\paragraph{Format of block address at depth $d$.}
Suppose that a block's logical 
address is a $\log_2 N$-bit string denoted by $\addrd{D} := \addr[1..(\log_2 N)]$ (expressed
in binary format), where $\addr[1]$ is the most significant bit.
In general, at depth~$d$, an address $\addrd{d}$ is the length-$d$ 
prefix of the full address $\addrd{D}$.
Henceforth, we refer to $\addrd{d}$ as a depth-$d$ address
(or the depth-$d$ truncation of $\addr$).

When we look up a data block, we would look up 
the full address $\addrd{D}$ in recursion depth $D$; 
we look up $\addrd{D-1}$
at depth $D-1$, 
$\addrd{D-2}$ at depth 
$D-2$, and so on. Finally at depth $0$, 
only one block is stored at $\ORAM_0$.

{\it A block with the address $\addrd{d}$
in $\ORAM_d$ 
stores the position labels for two 
blocks in $\ORAM_{d+1}$, 
at addresses $\addrd{d} || 0$ 
and 
$\addrd{d} || 1$  respectively.}
Henceforth, we say that the two addresses
$\addrd{d} || 0$ and $\addrd{d} || 1$ 
are {\it siblings} to each other; $\addrd{d} || 0$ 
is called the left sibling and 
$\addrd{d} || 1$ is called the right sibling.
We say that $\addrd{d} || 0$ 
is the left child of $\addrd{d}$ and 
$\addrd{d} || 1$ 
is the right child of $\addrd{d}$.

\subsubsection{An ORAM Lookup}
\label{sec:oram-lookup}

An ORAM lookup request is denoted as $(\op, \addr, \data)$ where
$\op \in \{\Read, \Write\}$. 
If  $\op = \Read$ then $\data := \bot$.
Here, $\addr$ denotes the address to lookup from
the ORAM.
The inputs are all provided by the client whereas the servers
store position-based $\ORAM_0, \ldots, \ORAM_D$ as discussed in
the previous section.
We perform the following operations:
\begin{enumerate}[leftmargin=3mm]
\item \textbf{Fetch.} For $d := 0$ to $D$, perform the following:
  \begin{itemize}[leftmargin=2mm]
  \item Let $\addrd{d}$ denote the depth-$d$ truncation of $\addrd{D}$.
  \item Call $\ORAM_d.{\sf Lookup}$ to lookup $\addrd{d}$. Recall that the position
    labels for the block will be obtained from the lookup of $\ORAM_{d-1}$.
    For $\ORAM_0$, no position label is needed.
		
	\item The block returned from ${\sf Lookup}$
	is placed in a special $\TOTM'_0$ in $\ORAM_d$. Jumping ahead, this will be merged
	with the rest of the data structure in the maintain phase.

  \item If $d < D$, each lookup will return two positions for addresses $\addrd{d}||0$
    and $\addrd{d}||1$. One of these will correspond to the position of $\addrd{d+1}$
    which will be required in the lookup for $\ORAM_{d+1}$.
  \item If $d = D$, the outcome of ${\sf Lookup}$ will contain the data block fetched.
  \end{itemize}

\item \textbf{Maintain.} We first consider depth $D$. Set depth-$D$'s update array
  $\Ud{D} := \emptyset$.
	
	Suppose $l^{D}$ is the smallest empty level in $\ORAM_D$.
  We have the invariant that for all $0 \leq d < D$, if $l^{D} < d$, then
  $l^{D}$ is also the smallest empty level in $\ORAM_d$.
  
  For $d := D$ to $0$, perform the following:
  \begin{enumerate}[leftmargin=2mm]
  \item If $d < l^{D}$, set $l := d$; otherwise, set $l := l^{D}$.
  \item Call $U \leftarrow \ORAM_d.{\sf Shuffle}((\OTM^d_0,
    \ldots, \OTM^d_l, \Ud{d}), l)$\anti{Shuffle should get as input the distributed layout}.
	
	Recall that to complete the description of ${\sf Shuffle}$,
			we need to specify the ${\sf update}$ rule
			that determines how to combine 
			the values of the same address that appears in both the current $\ORAM_d$
			and $\Ud{d}$.

		For $d  < D$, in $\Ud{d}$ and $\ORAM_{d}$, 
    each depth-$d$ logical address $\addrd{d}$ stores the  
    position labels for both children addresses $\addrd{d}||0$ and $\addrd{d}||1$ (in depth $d+1$).
		For each of the child addresses, if $\Ud{d}$ contains
    a new position label, choose the new one; otherwise, choose
    the old label previously in $\ORAM_{d-1}$.

  \item If $d \geq 1$, we need to send the updated positions
    involved in $U$ to
    depth $d-1$.
    We use the 
    ${\sf Convert}$ subroutine (detailed description below)  to convert
    $U$ into an update array for depth-$(d-1)$ addresses, where
    each entry may pack  
    the position labels for up to two sibling depth-$d$ addresses.
    
    Set $\Ud{d-1} \leftarrow {\sf Convert}(U, d)$,
    which will be used
    in the next iteration for recursion depth  $d-1$
    to perform its shuffle. 
  \end{enumerate}
\end{enumerate}

\subsubsection{The ${\sf Convert}$ subroutine.}
\label{sec:convertsubroutine}

$U$ is a sorted, unpermuted layout
representing the abstract array
$\{(\addrsd{i}{d}, \pos_i) : i \in [|U|]\}$.
The subroutine ${\sf Convert}(U, d)$ proceeds as follows.

For $i := 0$ to $|U|$, the client reconstructs
$(\addrsd{i-1}{d}, \pos_{i-1}), (\addrsd{i}{d}, \pos_{i})$ and
$(\addrsd{i+1}{d}, \pos_{i+1})$, computes $u'_i$ using the rules
below and 
secretly writes $u'_{i}$ to $U^{d-1}$.
\begin{itemize}[leftmargin=5mm]
\item 
If $\addrsd{i}{d} = \addr || 0$ 
and $\addrsd{i+1}{d} = \addr || 1$ for some $\addr$, i.e., if my right neighbor
is my sibling, then write down
$u'_i := (\addr, (\pos_i, \pos_{i+1}))$,
i.e., both siblings' positions need to be updated.

\item 
If $\addrsd{i-1}{d} = \addr || 0$ 
and $\addrsd{i}{d} = \addr || 1$ for some $\addr$, i.e., if my left neighbor
is my sibling, then write down
$u'_i := \bot$.

\item 
Else if $i$ does not have a neighboring sibling, 
parse $\addrsd{i}{d} = \addr || b$ for some $b \in \{0, 1\}$,
then write down 
$u'_i := (\addr, (\pos_i, *))$
if $b = 0$ or 
write down $u'_i := (\addr, (*, \pos_i))$ if $b = 1$.
In these cases, only the position of one of the siblings needs to be updated
in $\ORAM_{d-1}$.

\item 
Let $\Ud{d-1} := \{u'_i : i \in [|U|]\}$. 
Note here  that each entry of $\Ud{d-1}$ contains a depth-$(d-1)$
address of the form $\addr$,  as well as 
the update instructions for
two position labels of the depth-$d$ addresses $\addr||0$ 
and $\addr||1$
respectively. 

We emphasize that when $*$ appears, this means that the position of
the corresponding depth-$d$ address does not need to be updated in $\ORAM_{d-1}$.

\item Output $\Ud{d-1}$.
\end{itemize}

\begin{lemma}\label{final}
The above ORAM scheme is perfectly oblivious in the presence of a
semi-honest adversary corrupting a single server.
\end{lemma}

\begin{proof}
Observe that we adopt the framework
of building a recursive ORAM, where each depth
is a position-based ORAM.  Such a framework
has been used for 
perfectly secure ORAM with one-server~\cite{perfectopram} assuming perfect secure cryptography.

As mentioned in Lemma~\ref{lemma:posoram_obliv},
the major difference here is that the data is secretly shared between
three servers, thereby eliminating the need of perfectly secure cryptography.
In a fashion similar to~\cite{perfectopram},
the security of our recursive ORAM derives its security from the underlying position-based ORAMs,
because the overall view of the adversary is composed of the views corresponding
to the position-based ORAMs in all the recursion depths,
plus deterministic access pattern to pass data (which is secretly shared by independent randomness)
between successive recursion depths.
Since our position-based ORAM in each recursion depth is perfectly oblivious,
so is our overall recursive ORAM.
\end{proof}
\ignore{
The proof can be found in the supplemental
material, Section~\ref{sec:final}.
}

\begin{fact}
Each ORAM access takes an amortized bandwidth blowup of $O(\log^2 N)$.
\end{fact}

\begin{proof}
There are $D = O(\log N)$ position-based ORAMs in the construction.
For each $0 \leq d \leq D$, we analyze the
number of block operations associated with $\ORAM_d$.
From Fact~\ref{fact:lookup_pos}, the bandwidth required by the
fetch phase is $O(d)$.
From Fact~\ref{fact:shuffle_pos},
building the $\OTM$ at level $\ell$ takes a bandwidth of $O(2^\ell)$.
However, since an level-$\ell$ is rebuilt every $\Theta(2^\ell)$ ORAM accesses,
the amortized bandwidth to rebuild each level is $O(1)$.
Hence, the amortized bandwidth for the maintain phase is $O(d)$.
Therefore, the total amortized bandwidth per ORAM access is
$\sum_{d=1}^D d = O(\log^2 N)$, which is also the bandwidth blowup.
\end{proof}

Summarizing the above, we arrive at the following main theorem:

\begin{theorem}[Perfectly secure 3-server ORAM]
There exists a 3-server ORAM 
scheme that satisfies perfect correctness and 
perfect security 
against a semi-honest adversary in control of any single server,
and achieves $O(\log^2 N)$ amortized bandwidth blowup (where $N$
denotes the total number of logical blocks). 
\kartik{does this correspond to definition 2.3?}
\end{theorem}

Finally,  similar to existing works
that rely on the recursion technique~\cite{asiacrypt11,PathORAM},  
we can achieve better bandwidth blowup with larger block sizes:
suppose 
each data block is at least $\Omega(\log^2 N)$
in size, and we still set 
the position map blocks to be $O(\log N)$ bits long, 
then our scheme achieves $O(\log N)$
 bandwidth blowup.

\ignore{
\begin{corollary}
For block sizes $\Omega(\log^2 N)$, each ORAM access takes an
amortized $O(\log N)$ block operations.
\end{corollary}
}


\section*{Acknowledgments}
T-H. Hubert Chan was supported in part by the Hong Kong
RGC under grant~17200418. Jonathan Katz was supported
in part by NSF award~\#1563722. Kartik Nayak was
supported by a Google Ph.D.\ fellowship. Antigoni Polychroniadou
was supported by the Junior Simons Fellowship. Elaine Shi was
supported in part by NSF 
award~CNS-1601879, 
a Packard Fellowship, and
a DARPA Safeware grant (subcontractor under IBM).


\bibliographystyle{abbrv}
\bibliography{cache_obliv,refs,crypto,ref,another,hardware}

\newpage

\ignore{
\chapter*{Supplemental Materials}\label{supmaterial}

\section{Building Block: Unpermute}
\label{appendix:unpermute}
\subsubsection{Definition of $\sf Unpermute$.}
${\sf Unpermute}$ is a protocol that realizes an ideal functionality 
$\mcal{F}_{\rm unperm}$ 
as defined below.
Intuitively, this functionality 
reverses the effect of $\mcal{F}_{\rm perm}$.
It takes some permuted input layout,
and returns the corresponding unpermuted layout.
However, to avoid each server knowing its original permutation,
the contents of each entry 
needs to be secret-shared again.

\begin{itemize}[leftmargin=5mm]

\item $\{\iperm, \Arr'_b\}_{b \in \zthree}\leftarrow{\sf Unpermute (\{\pi_{b},\Arr_b\}_{b \in \zthree},\bot)}$:

\begin{itemize}[leftmargin=5mm]
\item 
{\it Input}: Let $\{\pi_b, \Arr_b\}_{b \in \zthree}$
 be the layout provided as input. (Recall that $\Arr_b$ and $(\pi_{b+1}, \Arr_{b+1})$ are
stored in server $\S_b$.) 


The arrays have the same length $|\Arr_0| = |\Arr_1| = |\Arr_2| = n$, for some $n$.
The client obtains $\bot$ as the input.

\item 
{\it Ideal functionality $\mcal{F}_{\rm unperm}$}:

Sample independently and uniformly random ${\Arr'}_0, {\Arr'}_1$ of length $n$.

Now, define $\Arr'_2 := \Arr'_0 \oplus \Arr'_1 \oplus (\oplus_{b} \pi_b^{-1}(\Arr_b))$,
i.e., $\oplus_b \Arr'_b  = \oplus_{b} \pi_b^{-1}(\Arr_b)$.


The output layout is 
$\{\iperm, \Arr'_b\}_{b \in \zthree}$,
and the client's 
output is $\bot$. 


\end{itemize}
\end{itemize}

\subsubsection{Protocol ${\sf Unpermute}$.}

The implementation of $\mcal{F}_{\rm unperm}$ proceeds as follows:
\begin{enumerate}[leftmargin=5mm]


\item \emph{Compute inverse permutations.} For each $b \in
  \zthree$, the client computes the inverse permutation
 $\widehat{\Arr}_{b+1} := \pi^{-1}_{b+1}(\Arr_{b+1})$ on server $\S_{b}$.

\item \emph{Mask shares.} For each data block, the client
  generates
block ``masks'' that sum up to \emph{zero} and then applies the
mask to $\Arr_{b+1}$ on server $\S_b$. 
Specifically, the client performs the following, for each $i \in
[n]$: 
\begin{compactitem}
\item Generate block ``masks'' that sum up to zero, i.e., Sample independent random blocks $\mathsf{B}^i_0$
and $\mathsf{B}^i_1$, and compute $\mathsf{B}^i_2 := \mathsf{B}^i_0 \oplus \mathsf{B}^i_1$.

\item Apply mask $\mathsf{B}^i_{b+1}$ to $\Arr_{b+1}[i]$ stored
  on server $\S_b$, i.e., for each $i \in [b]$, the client 
writes $\Arr'_{b+1}[i] \gets \widehat{\Arr}_{b+1}[i] \oplus \mathsf{B}^i_{b+1}$
on server $\S_b$.


\end{compactitem}

\item For each $b \in \zthree$, the server $\S_b$
sends $\Arr'_{b+1}$
to $\S_{b+1}$.

Hence, the new layout $\{\iperm, \Arr'_b\}_{b \in \zthree}$
is achieved.

\end{enumerate}

\ignore{
  \begin{lemma}[Perfectly Secure Implementation of $\mcal{F}_{\rm unperm}$]
\label{lemma:sec_unperm}
The above implementation of $\mcal{F}_{\rm unperm}$
is perfectly secure as per Definition~\ref{secdef}.
\end{lemma}
}

\begin{theorem}\label{thm:unpermute}
The ${\sf Unpermute}$ protocol perfectly securely realizes the
ideal functionality $\mcal{F}_{\rm unperm}$ in the presence of a
semi-honest adversary corrupting a single server with  $O(n)$
bandwidth blowup.

\end{theorem}

\begin{proof}
The proof is essentially the same as Theorem~\ref{thm:permute}.
It can be checked that
for each $b \in \zthree$,
the $\view^b$ is a deterministic function
of the inputs and the outputs of $\S_b$.
\end{proof}

\ignore{
\begin{fact}
  The ${\sf Unpermute}$ protocol
	takes $O(n)$ runtime,
	where each block has size $\Omega(\log n)$ bits.
\end{fact}
\begin{proof}
  Similar to the proof for Fact~\ref{fact:permute}.
\end{proof}
}

\section{Proofs}
\label{appendix:proofs}

\subsection{Proof of Theorem~\ref{thm:permute}}\label{ap:permute}

\ignore{\begin{lemma}[Perfectly Secure Implementation of $\mcal{F}_{\rm perm}$]
\label{lemma:sec_perm}
The above implementation of $\mcal{F}_{\rm perm}$
is perfectly secure as per Definition~\ref{secdef}.
\end{lemma}}

\begin{proof}
By construction, the implementation of the protocol
applies the correct permutation on each server's array
and re-distributes the secret shares using fresh independent randomness.
Hence, the marginal distribution of the protocol's outputs
is exactly the same as that of the ideal functionality.

Fix some $b \in \zthree$
and consider the $\view^b$ of the corrupt server $\S_b$.
Since the leakage of $\mcal{F}_{\rm perm}$ is empty,
we will in fact show that
given the inputs and the outputs to server $\S_b$,
the $\view^b$ is totally determined and has no more randomness.
Hence, given $\Inp_b$ and conditioning on $\Out_b$,
the $\view^b$ is trivially independent of the outputs
of the client and other servers.

Then, given the inputs $\Inp_b$ and the outputs $\Out_b$ to
$\S_b$, a simulator simply returns 
the view of $\S_b$ uniquely determined by $\Inp_b$ and $\Out_b$.


The inputs to $\S_b$ 
are the arrays $\Arr_b$ and $\Arr_{b+1}$,
and the permutation $\pi_{b+1}$.
The outputs are the arrays $\Arr'_b$ and $\Arr'_{b+1}$,
and also the permutation $\pi_{b+1}$.
We next consider each part of $\view^b$.

\begin{enumerate}

\item \textbf{Communication Pattern.}  The communication pattern between the client and all the servers only depends on~$n$.

\item \textbf{Data Structure.}  We next analyze the intermediate data that is observed by
$\S_b$.  The arrays $\Arr'_b$ and $\Arr'_{b+1}$ are in $\S_b$'s outputs.
Hence, it suffices to consider the intermediate
array $\widehat{\Arr}_{b+1} = \pi_{b+1}^{-1}(\Arr'_b)$,
which is totally determined by the outputs.
%
\end{enumerate}

We have shown that the $\view^b$
is actually a deterministic function of the inputs and the outputs of $\S_b$,
as required.
\paragraph{Efficiency. }   Recall that it takes $O(1)$ time to process one block.
	From the construction, it is straightforward that
	linear scans are performed on the relevant arrays.
			
	In particular, the two steps -- masking shares and permuting the
  shares -- can be done with $O(n)$ bandwidth. 
	Moreover, the client can generate a permutation on the
  server with $O(n)$ bandwidth, when each block has $\Omega(\log n)$ bits, using the
  Fisher-Yates algorithm~\cite{donald1998art}.

\end{proof}

\ignore{\begin{fact}
\label{fact:permute}
  The ${\sf Permute}$ protocol takes $O(n)$ runtime,
	where each block has size $\Omega(\log n)$ bits.
	
\end{fact}
\begin{proof}
  Recall that it takes $O(1)$ runtime to process one block.
	From the construction, it is straightforward that
	linear scans are performed on the relevant arrays.
			
	In particular, the two steps -- masking shares and permuting the
  shares -- can be done in $O(n)$ time. 
	Moreover, the client can generate a permutation on the
  server with $O(n)$ runtime, when each block has $\Omega(\log n)$ bits, using the
  Fisher-Yates algorithm~\cite{donald1998art}.
	\end{proof}}

\subsection{Proof of Theorem~\ref{thm:compaction}.}\label{ap:compaction}

\begin{proof}
By construction, the protocol correctly removes dummy elements and preserves
the original order of real elements, where the secret shares are re-distributed
using independent randomness.  Hence, the marginal distribution on the outputs
is the same for both the protocol and the ideal functionality.

We fix the inputs of all servers, and some $b \in \zthree$. 
The goal is to show that (1) the $\view^b$ follows a distribution
that is totally determined by the inputs $\mathbf{I}_b$ and the outputs $\mathbf{O}_b$ of the corrupt $\S_b$;
(2) conditioning on $\Out_b$,
$\view^b$ is independent of the outputs of the client and all other servers.

The second part is easy, because the inputs are fixed.
Hence, conditioning on $\Out_b$ (which includes $\Arr''_b$ and $\Arr''_{b+1}$),  
$\Arr''_{b+2}$ has no more randomness and totally determines the outputs of other servers.
%
%

To prove the first part, our strategy is to decompose $\view^b$ into a list components,
and show that fixing $\mathbf{I}_b$ and  conditioning on  $\mathbf{O}_b$ and a prefix 
of the components, the distribution of the next component can be determined.
Hence, this also gives the definition of a simulator.

First, observe that in the last step, the client
re-distributes the shares, and gives
output~$\mathbf{O}_b$ to $\S_b$;
moreover, the shares of $\Arr''$ are generated with fresh independent
randomness.  Hence, the distribution of the part of $\view^b$ excluding $\mathbf{O}_b$ is independent of $\mathbf{O}_b$.
We consider the components of $\view^b$ in the following order,
which is also how a simulator generates a view after seeing $\Inp_b$.

\begin{enumerate}

\item \emph{Communication Pattern.}
Observe that from the description of the algorithm,
the communication pattern between the client and the servers
depends only on the length~$n$ of the input array.

\item \emph{Random permutation $\pi_{b+1}$.}  This is independently generated
using fresh randomness.

\item \emph{Link $\link$ creation.}  The (abstract) array $\link$
is created by reverse linear scan.  The shares $\link_b$ and $\link_{b+1}$
received by $\S_b$ are generated by fresh independent randomness.\anti{again need to refer to the privacy of the sharing for the linked lists}

\item \emph{$\mathsf{Permute}$ subroutine.}  By the correctness
of the $\mathsf{Permute}$, the shares of the outputs $(\Arr', \link')$ 
received by $\S_b$ follow an independent and uniform random distribution.
By the perfect security of 
$\mathsf{Permute}$, the component of $\view^b$ due to
$\mathsf{Permute}$ depends only on the inputs (which include the shares
of $\Arr$ and $\link$) and the outputs of the subroutine.

\item \emph{List traversal.} Since $\S_b$ does not know $\pi_b$ (which is generated
using independent randomness by $\S_{b-1}$),
from $\S_b$'s point of view,
its array $\Arr'_b$ is traversed in an independent and uniform random order.
\end{enumerate}

Therefore, we have described a simulator procedure that
samples the $\view^b$ step-by-step, given $\mathbf{I}_b$
and $\mathbf{O}_b$.

\paragraph{Efficiency. }

  Each of the steps in the protocol can be executed with a
  bandwidth of
  $O(n)$. 
	Step~\ref{compact:step1} can be performed using
  Fisher-Yates shuffle algorithm. In steps~\ref{compact:step2} and
  \ref{compact:step4}, the client linearly scans the abstract
  lists $\Arr, \Arr''$ and the links $\link, \link'$. Accessing
  each array costs $O(n)$ bandwidth.
	Finally,
  step~\ref{compact:step3} invokes ${\sf permute}$, which
  requires $O(n)$ bandwidth
  (Theorem~\ref{thm:permute}).

\end{proof}

\ignore{\begin{fact}
  \label{fact:stable-compact}
  The ${\sf StableCompact}$ protocol
  has a runtime $O(n)$. 
	.
\end{fact}}

\subsection{Proof of Theorem~\ref{thm:merge}}
\label{proof:merge}
\begin{proof}
We follow the same strategy as in Theorem~\ref{thm:compaction}.
Again, from the construction, the protocol performs merging correctly and re-distributes secretes
using independent randomness.  Hence, the marginal distribution of
the outputs is the same for both the protocol and the ideal functionality.

We fix the inputs of all servers, and some $b \in \zthree$. 
Recall that the goal is to show that (1) the $\view^b$ follows a distribution
that is totally determined by the inputs $\mathbf{I}_b$ and the outputs $\mathbf{O}_b$ of $\S_b$;
(2) conditioning on $\Out_b$,
$\view^b$ is independent of the outputs of the client and all other servers.

The second part is easy, because the inputs are fixed.
Hence, conditioning on $\Out_b$ (which includes $\U''_b$ and $\U''_{b+1}$),  
$\U''_{b+2}$ has no more randomness and totally determines the outputs of other servers.
%
%

To prove the first part, our strategy is to decompose $\view^b$ into a list components,
and show that fixing $\mathbf{I}_b$ and  conditioning on  $\mathbf{O}_b$ and a prefix 
of the components, the distribution of the next component can be determined.
Hence, this also gives the definition of a simulator.

First, observe that in the last step, the client
re-distributes the shares of $\U''$, and gives
output~$\mathbf{O}_b$ (including $\U''_b$ and $\U''_{b+1}$) to $\S_b$;
moreover, the shares of $\U''$ are generated with fresh independent
randomness.  Hence, the distribution of the part of $\view^b$ excluding $\mathbf{O}_b$ is independent of $\mathbf{O}_b$.
We consider the components of $\view^b$ in the following order,
which is also how a simulator generates a view after seeing $\Inp_b$.

\begin{enumerate}

\item \emph{Communication Pattern.}
Observe that from the description of the algorithm,
the communication pattern between the client and the servers
depends only on the length~$n$ of the input arrays.

\item \emph{Random permutation $\pi_{b+1}$.}  This is independently generated
using fresh randomness.

\item \emph{Link $\link$ creation.}  The (abstract) array $\link$
is created by reverse linear scan.  The shares $\link_b$ and $\link_{b+1}$
received by $\S_b$ are generated by fresh independent randomness.

\item \emph{$\mathsf{Permute}$ subroutine.}  By the correctness
of the $\mathsf{Permute}$, the shares of the outputs $(\U', \link')$ 
received by $\S_b$ follow an independent and uniform random distribution.
By the perfect security of 
$\mathsf{Permute}$, the component of $\view^b$ due to
$\mathsf{Permute}$ depends only on the inputs (which include the shares
of $\Arr$ and $\link$) and the outputs of the subroutine.

\item \emph{List traversal.} Since $\S_b$ does not know $\pi_b$ (which is generated
using independent randomness by $\S_{b-1}$),
from $\S_b$'s point of view,
while the elements from the two underlying lists are being merged,
its array $\U'_b$ is traversed in an independent and uniform random order.
\end{enumerate}

Therefore, we have described a simulator procedure that
samples the $\view^b$ step-by-step, given $\mathbf{I}_b$
and $\mathbf{O}_b$.

\paragraph{Efficiency.} The analysis for the ${\sf Merge}$ protocol is similar to that for
  Theorem~\ref{thm:compaction} except that the operations are
  performed on lists of size $2n$ instead of $n$.

\end{proof}

\subsection{Proof of Lemma~\ref{lemma:44}}
\label{proof:OTM}
\begin{proof}
We go through each subroutine and explain its purpose,
from which correctness follows;
moreover, we also argue why it is perfectly
oblivious. 
Fix some $b \in \zthree$, and we consider the $\view^b$ of corrupt $\S_b$.

\noindent \textbf{Build subroutine.}  We explain each step, and why
the corresponding $\view^b$ satisfies perfect obliviousness.

\begin{enumerate}
\item \emph{Initialize to add dummies.} 
This step appends $n$ dummies at the end of the abstract array.
The data access pattern depends only $n$.  The contents of the extra dummy entries
are secretly shared among the servers using independent randomness.

\item \emph{Generate permutations for $\OTM$.} 
The client generates an independent random permutation for each server.
In particular, the permutation $\pi_{b+1}$ received by $\S_b$ is independent and uniformly random.

\item \emph{Construct a dummy linked list.} 
In this step, a linked list is created for the dummy entries.
Using $\pi_{b+1}$, server~$\S_b$ observes a linear scan on its array
and creates a linked list that is meant for $\S_{b+1}$.
Similarly, the linked list for $\S_b$ is created by $\S_{b-1}$.

The links themselves are secretly shared and stored in the abstract $\link$,
and so are the head positions~$\oplus_b {\sf dpos}_b$.
Hence, the data seen by $\S_b$ appears like independent randomness.

\item \emph{Construct the key-position map $U$}. 
This step is just a linear scan on the abstract $\Arr'$ and
each permutation array $\pi_b$.
The resulting abstract array $U$ is also secretly shared.

\item \emph{Permute the lists along with the links.}
This steps uses the building block ${\sf Permute}$,
which is proved to be perfectly secure in Lemma~\ref{thm:permute}.

\item Finally, as guaranteed by ${\sf Permute}$,
the resulting arrays $\widehat{\Arr}$ and $\widehat{\link}$ are secretly shared 
using independent randomness.  Moreover, the abstract list $U$
is returned as required.  Hence, the ${\sf Build}$ subroutine 
is correct and perfectly oblivious, as required.

\end{enumerate}

\noindent \textbf{Lookup Subroutine.} By construction,
each call to ${\sf Lookup}$ is supplied with the correct position labels,
and hence, correctness is achieved.  We next argue
why the distribution of the $\view^b$ of corrupt $\S_b$ depends only on
the number $\ell$ of lookups.

Observe that the entries accessed in $\S_b$ are permuted randomly by $\pi_b$, which is unknown to
$\S_b$.  Since the request sequence is non-recurrent, each real key is requested at most once.
On the other hand, the dummy entries are singly linked, and each dummy entry is accessed also at most once.
Hence, the $\ell$ lookups correspond to $\ell$ distinct uniformly random accesses in $\S_b$'s corresponding arrays.

\noindent \textbf{Getall Subroutine.}  This subroutine calls the building block ${\sf Unpermute}$,
which is proved to be perfectly secure in Lemma~\ref{thm:unpermute}.

As mentioned before, all data stored on $\S_b$ is secretly shared using independent randomness.
Hence, the distribution of the overall $\view^b$ depends only on the number $\ell$ of lookups.

\end{proof}

\subsection{Analysis of Position-based ORAM}\label{sec:ana}

\begin{fact}[Correctness]
Our Position-based ORAM maintains correctness.
More specifically, at every recursion depth $d$, the correct position
labels will be input to the ${\sf Lookup}$ operations
of $\ORAM_d$; and every request
will return the correct answers.
\end{fact}
\begin{proof}
Straightforward by construction.
\end{proof}

For every $\ORAM_d$ at recursion depth $d$, 
the following invariant are respected by construction
as stated in the following facts. 

\begin{fact}[Lifetime of $\OTM$]
\label{fct:otmlife}
For every $\ORAM_d$, every $\OTM_i$ instance at level $i \leq d$
that is created needs to answer 
at most $2^i$ requests 
before $\OTM_i$ instance is destroyed. 
\end{fact}
\begin{proof}
For every $\ORAM_d$, the following is true:
imagine that there is a $(d+1)$-bit binary counter initialized to $0$ that increments
whenever a request comes in.
Now, for $0 \leq \ell < d$, 
whenever the $\ell$-th bit flips from $1$ to $0$, the $\ell$-th level 
of $\OTM_\ell$ is destroyed;
whenever the $\ell$-th bit flips from $0$ to $1$, the $\ell$-th level
of $\OTM_\ell$ is reconstructed. 
For the largest level $d$ of $\ORAM_d$, 
whenever the $d$-th (most significant)
bit of this binary counter flips from $0$ to $1$ or from $1$ to $0$, 
the $(d+1)$-st level is destroyed and reconstructed.
The fact follows in a straightforward manner by observing 
this binary-counter argument.
\end{proof}

\begin{fact}\label{fct:otmnonrecur}
For every $\ORAM_d$ and every $\OTM_\ell$ instance at level $\ell \leq d$, 
during the lifetime of the $\OTM_\ell$ instance:
(a) no two real requests 
will ask for the same key;
and (b) for every request that asks for a real $\key$,
a block with $\key$ must exist in $\OTM_i$.
\end{fact}
\begin{proof}
We first prove claim (a).
Observe that for any $\ORAM_d$, if the block with some 
$\key$
is fetched from 
some level $\ell \leq d$, 
at this moment,
that block
will either enter a smaller level $\ell' < \ell$;
or some level $\ell'' \geq \ell$ will be rebuilt and
the block with $\key$
will go into level $\ell''$ --- in the latter case,
level $\ell$ will be destroyed prior to the rebuilding of level $\ell''$. 
In either of the above cases, 
due to correctness of the construction, if the  
$\key$
is needed again from $\ORAM_d$, 
a correct position label will be 
provided for 
$\key$
such that the request 
will not go to level $\ell$ (until the level is reconstructed).
Finally, claim (b) follows from correctness of the position labels.
\end{proof}

Given the above facts, our construction maintains perfect obliviousness
as in~\cite{perfectopram}.  We emphasize that while the proof structure
is the same, our notion of perfect obliviousness is defined in terms
of the view of one server plus the communication pattern between the client and all the servers.

\begin{lemma}[Obliviousness]
\label{lemma:posoram_obliv}
The above position-based ORAM construction satisfies perfect obliviousness
in the presence of a
semi-honest adversary corrupting a single server.
\end{lemma}

\begin{proof}
We would like to point out the major difference from~\cite{perfectopram}.
Since we do not rely on cryptographic assumptions,
the data content is protected by secret-sharing among the three servers.

For every $\OTM$ instance
constructed during the lifetime of the ORAM,  
Facts~\ref{fct:otmlife}
and \ref{fct:otmnonrecur} are satisfied, and thus every 
one-time memory instance
receives a valid request sequence.
The lemma then follows 
in a straightforward fashion by the perfect obliviousness
of the $\OTM$ scheme (and the perfect security of the underlying building blocks such as ${\sf Merge}$), 
and by observing
that all other access patterns of the ORAM construction are
deterministic
and independent of the input requests.
\end{proof}

%
%
%

\begin{fact}
\label{fact:shuffle_pos}
  Suppose that the {\sf update} function can be evaluated in
  $O(1)$ time.
	Then, for an
  $\ORAM_d$, and $l \leq d$, ${\sf Shuffle}_d$ for level $l$
  requires a bandwidth of
  $O(2^l)$.
\end{fact}
\begin{proof}
  Observe that to obtain real key-value pairs from $\TOTM$'s and
  to construct a level $l$ list, the client invokes {\sf Getall}
  and {\sf Merge} protocols on (abstract) lists of sizes $2^0,
  \ldots, 2^l$, and {\sf StableCompact} on size $O(2^l)$. Since
  each of these steps require linear number of block operations,
	the total number of block operations for these
  steps is $\sum_{i = 1}^{l} (2^i) + O(2^l) = O(2^l)$.
  In order to update the list to contain values from $U$, we
  again invoke {\sf Merge}, {\sf StableCompact}, and an {\sf
    update}  on a list of
  size $O(2^l)$. This requires $O(2^l)$ block operations
  assuming the {\sf update} rule itself requires $O(1)$ operations
  per block.
  Finally, the call to ${\sf Build}$ $\OTM_l$ requires $O(2^l)$
  block operations.
\end{proof}

\subsection{Proof of Lemma \ref{final}}\label{sec:final}

\begin{proof}
Observe that we adopt the framework
of building a recursive ORAM, where each depth
is a position-based ORAM.  Such a framework
has been used for 
perfectly secure ORAM with one-server~\cite{perfectopram} assuming perfect secure cryptography.

As mentioned in Lemma~\ref{lemma:posoram_obliv},
the major difference here is that the data is secretly shared between
three servers, thereby eliminating the need of perfectly secure cryptography.
In a fashion similar to~\cite{perfectopram},
the security of our recursive ORAM derives its security from the underlying position-based ORAMs,
because the overall view of the adversary is composed of the views corresponding
to the position-based ORAMs in all the recursion depths,
plus deterministic access pattern to pass data (which is secretly shared by independent randomness)
between successive recursion depths.
Since our position-based ORAM in each recursion depth is perfectly oblivious,
so is our overall recursive ORAM.
\end{proof}

\begin{figure}[!tbp]

  \tikzset{
  triangle/.style = {draw=black, fill=brown!5, thick},
  squiggle/.style = {decoration={snake, segment length=5mm}, decorate}
}

\begin{tikzpicture}

   \centering
  \node[shape=regular polygon,regular polygon sides=3,draw=black, fill=brown!5, thick,
minimum size=5cm,anchor=south,pin=140:$\ORAM_D$] (A1) at (0,4.3) {};
\node[shape=regular polygon,regular polygon sides=3,draw=black, fill=brown!5, thick,
minimum size=4cm,anchor=south,pin=140:$\ORAM_{D-1}$] (A2) at (0,3.2) [above]{};
\node[shape=regular polygon,regular polygon sides=3,draw=black, fill=brown!5, thick,
minimum size=3cm,anchor=south,pin=140:$\ORAM_{D-2}$] (A3) at (0,2.1) [above]{};
\node[shape=regular polygon,regular polygon sides=3,draw=black, fill=brown!5, thick,
minimum size=0.5cm,anchor=south,pin=140:$\ORAM_0$] (A4) at (0,0.5) [above]{};
  

 \node[align=center] (n0)  at (0,2.2)  {$$};
 \node[align=center] (n1)  at (0,1.3)  {$$};
 \path (n0) -- (n1) node [blue, midway, sloped] {$\ldots$};

\draw[-{Latex[open]}] (2.5,4) -- ++(2,0);

  \node[shape=regular polygon,regular polygon sides=3,draw=black, fill=brown!5, thick,
minimum size=4cm,anchor=south,label=above:$\ORAM_{D-1}$] (B2) at (6,3.2) [above]{};
  
\draw (4.4,3.5) -- (7.6,3.5);
 \node[pin=0:$\OTM_{D-1}$] (a1)  at (7.55,3.3)  {$$};

\draw (4.6,3.8) -- (7.4,3.8);
 \node[pin=0:$\OTM_{D-2}$] (a2)  at (7.3,3.7)  {$$};

\draw (5.8,5.9) -- (6.2,5.9);
 \node[pin=0:$\OTM_0$] (a3)  at (6,6)  {$$};

 \node[align=center] (n2)  at (6,3.5)  {$$};
 \node[align=center] (n3)  at (6,5.5)  {$$};
 \path (n2) -- (n3) node [blue, midway, sloped] {$\ldots$};

\end{tikzpicture}

 \caption{ORAM Construction from Position-Based ORAM
 }\label{fig}
\end{figure}
}

\end{document}